\documentclass[journal,onecolumn]{IEEEtran}
\usepackage{etex}
\usepackage{mathtools}
\usepackage{ntheorem}
\usepackage{amssymb}

\usepackage{pictex,graphicx}

\usepackage{cite}

\usepackage{array}

\usepackage[tight,footnotesize]{subfigure}

\newtheorem{theorem}{Theorem}
\newtheorem{lemma}{Lemma}

\newtheorem{remark}{Remark}
\newtheorem{cor}{Corollary}

\begin{document}

\title{On Zero Delay Source-Channel Coding }

\author{Emrah~Akyol,~\IEEEmembership{Student Member,~IEEE,}
        Kumar~Viswanatha,~\IEEEmembership{Student Member,~IEEE,}
           Kenneth~Rose,~\IEEEmembership{Fellow,~IEEE,}
        and~Tor~Ramstad,~\IEEEmembership{Member,~IEEE}
\thanks{Emrah Akyol, Kumar Viswanatha and Kenneth Rose are with the Department
of Electrical and Computer Engineering, University of California, Santa Barbara,
CA, 93106 USA, e-mail: \{eakyol, kumar, rose\} @ece.ucsb.edu}
\thanks{ Tor Ramstad is with Dept. of Electronics and Telecommunications, Norwegian University of Science and Technology, Trondheim, Norway, email: ramstad@iet.ntnu.no}
\thanks{This work is supported in part by the NSF under the grants  CCF-0728986, CCF-1016861, CCF-1118075. The material in this paper was presented
in part at the IEEE Information Theory Workshop, Jan 2010 and the IEEE Data Compression Conference, March 2010. }}

\markboth{Submitted to IEEE Transactions  on Information Theory}%
{Akyol \MakeLowercase{\textit{et al.}}: On Zero Delay Source-Channel Coding }

\maketitle

\begin{abstract}
In this paper, we study the zero-delay source-channel coding problem, and specifically the problem of obtaining the vector transformations that optimally map between the $m$-dimensional source space and the $k$-dimensional channel space, under a given transmission power constraint and for the mean square error distortion.  We first study the functional properties of this problem and show that the objective is concave in the source and noise densities and convex in the density of the input to the channel. We then derive the necessary conditions for optimality of the encoder and decoder mappings.  A well known result in information theory pertains to the linearity of  optimal encoding and decoding mappings  in the scalar Gaussian source and channel setting, at all channel signal-to-noise ratios (CSNRs). In this paper,  we study this result more generally, beyond the Gaussian source and channel, and derive the necessary and sufficient condition for linearity of optimal mappings, given a noise (or source) distribution, and a specified power constraint. We also prove that the Gaussian source-channel pair is unique in the sense that it is the only source-channel pair for which the optimal mappings are linear at more than one CSNR values. Moreover, we show the asymptotic linearity of optimal mappings for low CSNR if the channel is Gaussian regardless of the source and, at the other extreme, for high CSNR if the source is Gaussian, regardless of the channel. Our numerical results show strict improvement over  prior methods. The numerical approach is extended to the scenario of source-channel coding with decoder side information. The resulting encoding mappings are shown to be continuous relatives of, and in fact subsume as special case, the Wyner-Ziv mappings encountered in digital distributed source coding systems. 

\end{abstract}

\begin{keywords}
Joint source channel coding, analog communications, estimation, distributed coding.
\end{keywords}

\IEEEpeerreviewmaketitle

\section{Introduction}
{A} fascinating result in information theory is that uncoded transmission of Gaussian samples, over a channel with additive white Gaussian noise (AWGN), is optimal in the sense that it yields the minimum achievable mean square error (MSE) between source and reconstruction \cite{goblick}. This result demonstrates the potential of joint source-channel coding: Such a simple scheme, at no delay, provides the performance of the asymptotically optimal separate source and channel coding system, without recourse to complex compression and channel coding schemes that require asymptotically long  delays. However, it is understood that the best source channel coding system at fixed  finite delay may not, in general, achieve Shannon's asymptotic coding bound (see e.g. \cite{coverbook}). 

Clearly, the problem of obtaining the optimal scheme for a given finite delay is an important open problem with considerable practical implications. There are two main approaches to the practical problem of transmitting a discrete time continuous alphabet source over a discrete time additive noise channel: ``analog communication" via direct amplitude modulation, and ``digital communication" which typically consists of quantization, error control coding and digital modulation. The main advantage (and hence proliferation) of digital over analog communication is due to advanced quantization and error control techniques, as well as the prevalence of digital processors. However, there are two notable shortcomings: First, error control coding (and to some extent also source coding) usually incurs substantial delay to achieve good performance. The other problem involves limited robustness of digital systems to varying channel conditions, due to underlying quantization or error protection assumptions. The performance saturates due to quantization as the channel signal to noise ratio (CSNR) increases beyond the regime for which the system was designed. Also, it is difficult to obtain ``graceful degradation" of digital systems with decreasing CSNR, when it falls below the minimum requirement of the error correction code in use. Further, such threshold effects become more pronounced as the system performance approaches the theoretical optimum. Analog systems offer the potential to avoid these problems. As an important example, in applications where significant delay is acceptable, a hybrid approach  (i.e., vector quantization + analog mapping) was proposed and analyzed \cite{mittal2002hybrid, skoglund2006hybrid} to circumvent the impact of CSNR mismatch where, for simplicity, linear mappings were used and hence no optimality claims made. Perhaps more importantly, in many applications delay is a paramount consideration. Analog coding schemes are low complexity alternatives to digital methods, providing a ``zero-delay" transmission which is suitable for such applications.

\begin{figure*}    \centering            \includegraphics[scale=0.35]{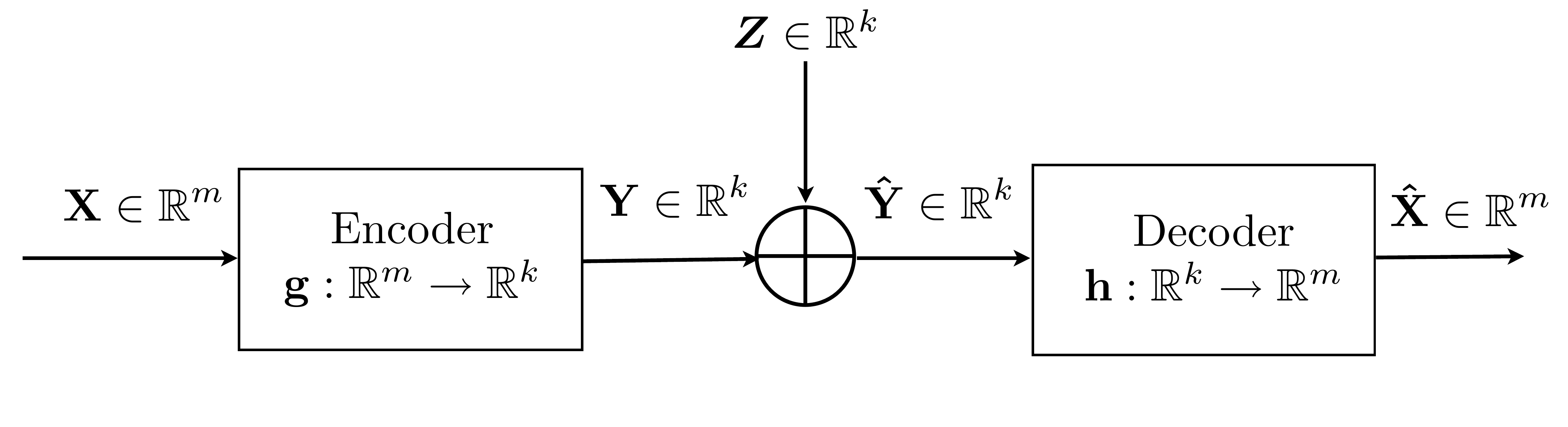}    \label{fig:pre2}    
\caption{A general block-based point-to-point communication system}
\end{figure*}

  There are no known explicit methods to obtain such analog mappings for a general source and channel, nor is the optimal mapping known, in closed form, for other than the trivial case of the scalar Gaussian source-channel pair. Among the few practical analog coding schemes that have appeared in the literature are those based on the use of space-filling curves for bandwidth compression, originally proposed more than 50 years ago by Shannon \cite{shannon1949cpn} and Kotelnikov \cite{kotelnikov}. These were then extended in the work of Fuldseth and Ramstad \cite{fuldseth}, Chung \cite{chung}, Vaishampayan and Costa \cite{vaishampayan2003curves}, Ramstad \cite{ramstad}, and Hekland et.al. \cite{hekland2009}, where spiral-like curves were explored for transmission of Gaussian sources over AWGN channels for bandwidth compression ($m>k$) and expansion ($m<k$). 
  
There exist two  main approaches to numerical optimization of the mappings. One is based on optimizing the parameter set of a structured mapping \cite{hu2011analog,wernersson2009polynomial,ramstad,hekland2009}. The performance of this approach is limited to the  parametric form (structure) assumed. For instance, Archimedian spiral based  space filling mappings  proposed   for bandwidth compression,  are known to perform well at high CSNRs.   The other is based on power constrained channel optimized vector quantization (PCCOVQ) where a ``discretized version" of the problem is tackled using tools developed for vector quantization \cite{fuldseth,floor2007power,karlsson2010optimized}.

  It is also noteworthy that a similar problem was solved in \cite{olc,basar1980performance} albeit under the stringent constraint that both encoder and decoder be linear. A related problem, formulated in the pure context of digital systems, was also studied by Fine \cite{fine1964properties}. Properties of the optimal mappings have been considered, over the years, in \cite{shannon1949cpn,ziv, trott}.  Shannon's arguments\cite{shannon1949cpn} are based on the topological impossibility to map between regions in a ``one-to-one", continuous manner, unless they have the same dimensionality. On this basis, he explained the threshold effect common to various communication systems. Moreover, Ziv \cite{ziv} showed that for a Gaussian source transmitted over AWGN channel, no single practical modulation scheme can achieve optimal performance at all noise levels, if the channel rate is greater than the source rate (i.e., bandwidth expansion). It has been conjectured that this result holds whenever the source rate differs from the channel rate \cite{trott}. Our own preliminary results appeared in \cite{emrah_itw10,emrah_dcc10}. 
  
 The existence of  optimal real time encoders have been studied in\cite{witsenhausen1979structure,walrand1983optimal,teneketzis2006structure} for encoding a Markov source with zero-delay. Along these lines, for similar set of problems, \cite{walrand1983optimal} demonstrated the existence the optimal causal encoders  using dynamic programming, its results are recently extended to partially observed Markov sources and multiterminal settings in \cite{yuksel2010optimal}. The problem we consider is intrinsically connected to problems in stochastic control where the controllers must operate at zero delay. A control problem, similar to the zero-delay source channel coding problem here, is the Witsenhausen's well known counterexample \cite{witsenhausen1968counterexample} (see e.g. \cite{basar2008variations} for a comprehensive review) where a similar functional optimization problem is studied and it is shown that nonlinear controllers can outperform linear ones in decentralized control settings even under Gaussianity and MSE assumptions. 

In this paper, we investigate the problem of obtaining vector transformations that optimally map between the $m$-dimensional source space and the $k$-dimensional channel space, under a given transmission power constraint, and where optimality is in the sense of minimum mean square reconstruction error. We provide necessary conditions for optimality of the mappings used at the encoder and the decoder. It is important to note that virtually any source-channel communication system (including digital communication) is a special case of the general mappings shown in Figure 1. A typical digital system, including quantization, error correction and modulation, boils down to a specific mapping from the source space $\mathbb  R^m$ to the channel space $\mathbb  R^k$ and back to reconstruction space $\mathbb  R^m$ at the receiver. Hence the derived optimality conditions are generally valid and subsume digital communications as an extreme special case. Based on the optimality conditions we derive, we propose an iterative algorithm to optimize the mappings for any given $m,k$ (i.e., for both bandwidth expansion or compression) and for any given source-channel statistics. To our knowledge, this problem has not been fully solved, except when both source and channel are scalar and Gaussian. We provide examples of such $m:k$ mappings for source-channel pairs and construct the corresponding source-channel coding systems that outperform the mappings obtained in \cite{fuldseth, chung, vaishampayan2003curves,ramstad, hekland2009}. 

We also study the functional properties of the zero delay source-channel coding problem. Specifically, first we show that the end-to-end mean square error is a concave functional of the source density given  fixed noise density, and of the noise density given a fixed source. Secondly, for the scalar version of the problem, the minimum mean square error is a convex  functional of the channel input density. The convexity result makes the optimal encoding mapping ``essentially unique" \footnote{The optimal mapping is not strictly unique, in the sense that multiple trivially ``equivalent" mappings can be used to obtain the same channel input density. For example, a scalar unit variance Gaussian source and scalar Gaussian channel with power constraint $P$, can be optimally encoded by either  $y=\sqrt P x$ or $y=-\sqrt P x$.} and paves the way to the linearity results we present later.

In digital communications, structured (linear, lattice etc.) codes (mappings) are generally optimal, i.e., the structure comes without any loss in performance. We observe that this is not the case at zero-delay. Hence, characterizing the conditions, in terms of the source, channel densities and the power constraint, that yield linearity of optimal mappings is an interesting open problem. Building on the functional properties we derive in this paper, and the recent results on conditions for linearity of optimal estimation \cite{emrah_estimation}, we derive the necessary and sufficient conditions for linearity of optimal mappings. We then study the CSNR asymptotics and particularly show that given a Gaussian source, optimal mappings are asymptotically linear at high CSNR, irrespective of the channel. Similarly, for a Gaussian channel, optimal mappings are asymptotically linear at low CSNR regardless of the source.


The last part of the paper extends the numerical approach to the scenario of source-channel coding with decoder side information (i.e., the decoder has access to side information that is correlated with the source). This setting, in the context of pure source coding, goes back to the pioneering work of Slepian and Wolf \cite{slepianwolf} and Wyner and Ziv \cite{wynerziv}. The derivation of the optimality conditions for  the decoder side information setting is a direct extension of the point-to-point case, but the distributed nature of this setting results in highly nontrivial mappings. Straightforward numerical optimization of such mappings is susceptible to get trapped in numerous poor local minima that riddle the cost functional. Note, in particular, that in the case of Gaussian sources and channels, linear encoders and decoder (automatically) satisfy the necessary conditions of optimality while, as we will see, careful optimization obtains considerably better mappings that are far from linear.

In Section II, we formulate the problem.  In Section III,  we study the functional properties of the problem and derive the necessary conditions for optimality of the mappings. We then provide an iterative algorithm based on these necessary conditions, in Section IV. In Section V, we derive the necessary and sufficient conditions for linearity of encoding and/or decoding mappings, in the particular instance of scalar mappings. We provide example mappings and comparative numerical results in Section VI. Discussion and future work are presented in Section VII.

\section{Problem Formulation}

\subsection {Preliminaries and Problem Definitions}

Let $\mathbb R$, $\mathbb N$, $\mathbb R^+$, and $\mathbb C$  denote the respective sets of real numbers, natural numbers, positive real numbers and complex numbers. In
general, lowercase letters (e.g., $x$) denote scalars, boldface lowercase (e.g., $\boldsymbol x$) vectors, upper-
case (e.g., $U, X$) matrices and random variables, and boldface uppercase (e.g., $\boldsymbol X$) random
vectors.  $\mathbb E(\cdot)$ and  $\mathbb P(\cdot)$ denote the expectation and probability operators, respectively. $||\cdot||$ denotes the $l_2$ norm. $\nabla$ denotes the gradient and $\nabla_x$ denotes the partial gradient with respect to $\boldsymbol x$. $f^{'} (\cdot)$ denotes the first order derivative of the function $f(\cdot)$, i.e., $f^{'} (x)= \frac{d f(x)}{d x}$. All the logarithms in the paper are natural logarithms and may in general be complex. The integrals are in general Lebesgue integrals. 

Let ${\cal {S}}_{m}^{k}$ denote the set of Borel measurable, square integrable  functions $\{ {\boldsymbol f}: \mathbb R^m \rightarrow \mathbb R^k$\}. Let us define the set $\cal S^+$ as the set of monotonically increasing deterministic, Borel measurable $\mathbb R \rightarrow \mathbb R$ mappings. 

\subsubsection{Point to point}
We consider the general communication system whose block diagram is shown in Figure 1. An $m$-dimensional zero mean\footnote{The zero mean assumption is not necessary, but it considerably simpliÞes the notation. Therefore, it is made throughout the paper.} vector source ${\boldsymbol X} \in\mathbb R^m $ is mapped into a $k$-dimensional vector ${\boldsymbol Y}\in  \mathbb R^k$ by function ${\boldsymbol g}\in {{\cal {S}}_{m}^{k}}$, and transmitted over an additive noise channel. The received vector $\boldsymbol {\hat Y}=\boldsymbol Y+\boldsymbol Z $ is  mapped by the decoder to the estimate $\boldsymbol {\hat X}$ via function ${\boldsymbol h} \in {{\cal {S}}_{k}^{m}}$. The zero mean  noise $\boldsymbol Z$ is assumed to be independent of the source $\boldsymbol X$. The $m$-fold source density is denoted $f_X(\cdot)$ and the $k$-fold noise density is $f_Z(\cdot)$ with characteristic functions $F_X(\boldsymbol \omega)$ and $F_Z(\boldsymbol \omega)$, respectively.

The objective is to minimize, over the choice of encoder ${\boldsymbol g} \in {{\cal {S}}_{m}^{k}}$ and decoder ${\boldsymbol h} \in {{\cal {S}}_{k}^{m}}$, the distortion
\begin{equation}
D [{\boldsymbol g,\boldsymbol h }] ={\mathbb E}\{||{\boldsymbol X}-{\boldsymbol{ \hat X}}||^2\},
\end{equation}
subject to the average power constraint,
\begin{equation}
P [{\boldsymbol g}] = \mathbb E \{ ||{\boldsymbol g}({\boldsymbol X})||^2\}  \leq {P_T} \label{power_cons},
\end{equation}
\noindent where $P_T$ is the specified transmission power level. Bandwidth compression-expansion is determined by the setting of the source and channel dimensions, $k/m$. The power constraint limits the choice of encoder function $\boldsymbol g(\cdot)$. Note that without a power constraint on $\boldsymbol g(\cdot)$, the CSNR is unbounded and the channel can be made effectively noise free.

Our goal is to minimize MSE subject to the average power constraint. Let us write MSE explicitly as a functional of $\boldsymbol g(\cdot)$ and $\boldsymbol h(\cdot)$ 
\begin{align}
 D(\boldsymbol g, \boldsymbol h) =  \int   \int{ [\boldsymbol x- \boldsymbol h(\boldsymbol g(\boldsymbol x)+\boldsymbol z )]^T \boldsymbol [\boldsymbol x -\boldsymbol h(\boldsymbol g(\boldsymbol x)+\boldsymbol z)] }  f_X({\boldsymbol x}) f_Z({\boldsymbol z})  \mathrm{d}\boldsymbol x  \mathrm{d}\boldsymbol z 
 \end{align}
 To impose the power constraint, we construct the Lagrangian cost functional: 
\begin{equation}
J (\boldsymbol g, \boldsymbol h) =D(\boldsymbol g, \boldsymbol h)+\lambda \{P (\boldsymbol g)\}
\label{total_cost}
 \end{equation}
to minimize over the mapping $\boldsymbol g (\cdot)$ and $\boldsymbol h (\cdot)$. 

\subsubsection{Decoder side information}
As shown in Figure 2, there are two correlated vector sources ${\boldsymbol X_1} \in  \mathbb R^{m_1} $ and  ${\boldsymbol X_2} \in \mathbb R^{m_2} $ with a joint density $f_{X_1,X_2}(\boldsymbol x_1,\boldsymbol x_2)$.  $\boldsymbol X_2$ is available only to the decoder, while $\boldsymbol X_1$ is mapped to ${\boldsymbol Y} \in  \mathbb R^k$ by the encoding function ${\boldsymbol g} \in{{\cal {S}}_{m_1}^{k}}$ and transmitted over the channel whose additive noise ${\boldsymbol Z} \in \mathbb R^k$, with density $f_Z(\cdot)$, is independent of $\boldsymbol X_1, \boldsymbol X_2$. The received channel output $\boldsymbol{ \hat Y}=\boldsymbol Y+\boldsymbol Z $ is mapped to the estimate $\boldsymbol{ \hat X_1}$ by the decoding function ${\boldsymbol h} :  \mathbb R^k \times \mathbb R^{m_2}  \rightarrow  \mathbb R^m_1$. The problem is to find optimal mapping functions $\boldsymbol g, \boldsymbol h$ that minimize the distortion 

\begin{equation} 
D(\boldsymbol g, \boldsymbol h)={\mathbb E}\{||\boldsymbol X_1-\boldsymbol {\hat X_1}||^2\}
 \end{equation} 
   subject to average power constraint identical to (\ref{power_cons}).

\begin{figure*}
	\centering
			\includegraphics[scale=0.30]{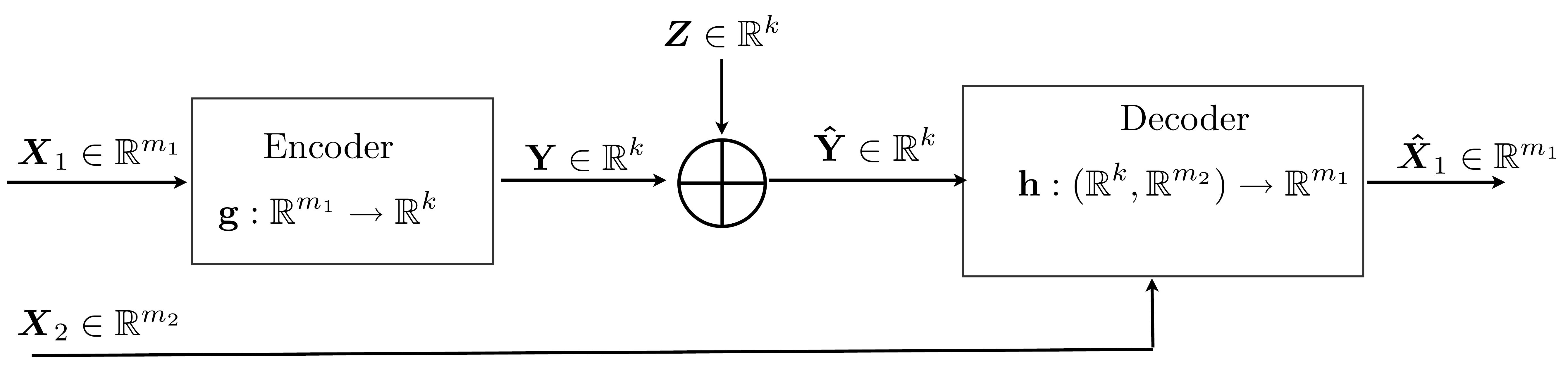}
	\caption{Source-channel coding with decoder side information}
\end{figure*}

\subsection {Asymptotic Bounds for Gaussian Source and Channel}
Although the problem we consider is delay limited, it is insightful to consider asymptotic bounds obtained at infinite delay. From Shannon's source and channel coding theorems, it is known that, asymptotically, the source can be compressed to $R(D)$ bits (per source sample) at distortion level $D$, and that $C$ bits can be transmitted over the channel (per channel use) with arbitrarily low probability of error, where $R(D)$ is the source rate-distortion function, and $C$ is the channel capacity, (see e.g.\cite{coverbook}). The asymptotically optimal coding scheme is the tandem combination of the optimal source and channel coding schemes,  hence $m R(D) \leq k C$ must hold. By setting 
\begin{equation}
\label{eq1}
R(D) =\frac{k}{m}C,
\end{equation}
one obtains a lower bound on the distortion of any source-channel coding scheme. Next, we 
specialize to Gaussian sources and channels, which we will mostly use in the numerical results
section, while emphasizing that the proposed method is generally applicable to any source and noise densities. The rate-distortion function for the memoryless Gaussian source of variance $\sigma_X^2$, under the squared-error distortion measure is given by
\begin{equation}
\label{eq2}
R(D)=\max(0,\frac{1}{2}\log\frac {\sigma_X^2}{D}),
\end{equation}
for any distortion value $D \geq 0$. The capacity of the AWGN channel is given by
\begin{equation}
\label{eq3}
C=\frac{1}{2}\log(1+\frac{P_T}{\sigma_Z^2}),
\end{equation}
where $P_T$ is the transmission power constraint and $\sigma_Z^2$ is the noise variance. Plugging (\ref{eq2}) and (\ref{eq3}) in (\ref{eq1}), we obtain the optimal performance theoretically attainable (OPTA):
\begin{equation}
\label{opta}
D_{OPTA}=\frac{\sigma_X^2}{(1+\frac{P_T}{\sigma_Z^2})^{\frac{k}{m}}}.
\end{equation}

For source coding with decoder side information, it has been established for Gaussians and MSE distortion that there is no rate loss due to the fact that the side information is unavailable to the encoder \cite{wynerziv}. Similar to the derivation above, OPTA can be obtained for source-channel coding with decoder side information, by equating the conditional rate distortion function of the source (given the side information) to the channel capacity. The rate distortion function of $X_1$ when $X_2$ serves as side information and $[X_1,X_2] \sim \mathcal N(\boldsymbol 0, R_X)$ where $R_X=\sigma_X^2 \left[ \begin{array}{cc}  1 & \rho \\ \rho &1 \end{array} \right ]$ with $|\rho|\leq 1$ is:
\begin{equation}
\label{eq22}
R(D)=\max(0,\frac{1}{2}\log\frac {(1-\rho^2)\sigma_X^2}{D}),
\end{equation}
Similar to point-to-point setting, we plug (\ref{eq22}) and (\ref{eq3}) in (\ref{eq1}) to obtain OPTA
\begin{equation}
\label{opta2}
D_{OPTA}=\frac{(1-\rho^2)\sigma_X^2}{(1+\frac{P_T}{\sigma_Z^2})^{\frac{k}{m}}}.
\end{equation}

Note that OPTA is derived without any delay constraints and may not be achievable by a delay-constrained coding scheme. No achievable theoretical bound is known for joint source channel coding with limited delay, although there are recent results that tighten the outer bound, see e.g. \cite{dlb, tridenski2011bounds, reani2012data}.

\section{Functional Properties of Zero-Delay Source-Channel Coding Problem}
In this section, we study the functional properties of the optimal zero-delay source-channel coding problem. These properties are not only important in their own right, but also, as we will show in the following sections, enable the derivation of several subsequent results.  

Let us restate the Lagrangian cost (\ref{total_cost}), as $J(\boldsymbol X, \boldsymbol Z, \boldsymbol g,\boldsymbol h)$ which makes explicit its dependence on the source and channel noise $\boldsymbol X \sim f_X(\cdot)$ and $\boldsymbol Z\sim f_Z(\cdot)$,  beside the deterministic mappings $\boldsymbol g(\cdot)$ and $\boldsymbol h(\cdot)$ as: 
\begin{equation}
J(\boldsymbol X, \boldsymbol Z, \boldsymbol g, \boldsymbol h)=\mathbb E \left \{ ||\boldsymbol X-\boldsymbol h(\boldsymbol g(\boldsymbol X)+\boldsymbol Z) ||^2 \right \} +\lambda \mathbb E \left \{ ||\boldsymbol g(\boldsymbol X)||^2 \right \}
\end{equation}
 The minimum achievable cost is
\begin{equation}
J_m(\boldsymbol X,\boldsymbol Z)\triangleq\inf_{\substack{\boldsymbol g,\boldsymbol h}}  J(\boldsymbol X,\boldsymbol Z, \boldsymbol g, \boldsymbol h)
\end{equation}
Similarly, conditioned on another random variable $\boldsymbol U$, $J_m( \boldsymbol X, \boldsymbol  Z|\boldsymbol  U)$ denotes the overall cost when $\boldsymbol U$ is available to both encoder and decoder. We define $J_r$ as the value of overall cost as a function of $\boldsymbol g(\cdot), $ when $ \boldsymbol h(\cdot)$ had already been optimized for $ \boldsymbol g(\cdot)$: 
\begin{equation}
J_r(\boldsymbol X,\boldsymbol Z,\boldsymbol g) \triangleq\inf \limits_{\boldsymbol h} J(\boldsymbol X, \boldsymbol Z, \boldsymbol g, \boldsymbol h)
\end{equation}

\subsection{Concavity of $J_m$ in $f_X(\cdot)$ and $f_Z(\cdot)$}
In this section, we show the concavity of the minimum cost, $J_m$ in the source density $f_X(\cdot)$ and in the channel noise density $f_Z(\cdot)$. Similar results were derived for the MMSE estimation in the scalar setting, in \cite{wu2012functional}, where no encoder is  present in the problem formulation. We start with the following lemma which states the impact of conditioning on the overall cost. 

\begin{lemma}
Conditioning cannot increase the overall cost, $J_m$ i.e., $J_m(\boldsymbol X, \boldsymbol Z)\geq J_m(\boldsymbol X, \boldsymbol Z|\boldsymbol U)$ for any $\boldsymbol U$. 
\label{cond1}
\end{lemma}
\begin{proof}
The knowledge of $\boldsymbol U$ cannot increase the total cost, since we can always ignore $U$ and use the $\boldsymbol g(\cdot),\boldsymbol h(\cdot)$ pair that is optimal for $J_m(\boldsymbol X,\boldsymbol Z)$. Hence,  $J_m(\boldsymbol X, \boldsymbol Z|\boldsymbol U)\leq J_m(\boldsymbol X, \boldsymbol Z)$.  
\end{proof}
Using Lemma \ref{cond1}, we prove the following theorem which states the concavity of the minimum cost $J_m(\boldsymbol X, \boldsymbol Z)$.  
\begin{theorem}
$J_m$ is concave in $f_X(\cdot)$ and $f_Z(\cdot)$.
\end{theorem}
\begin{proof}
Let  $\boldsymbol X$ be distributed according to $f_X=pf_{X_1}+(1-p)f_{X_2}$, where $f_{X_1}$ and $f_{X_2}$ respectively denote the densities of random variables $\boldsymbol X_1$ and $\boldsymbol X_2$.  Then, $\boldsymbol X$ can be expressed, in terms of a time sharing random variable $U$ which takes values in the alphabet $\{1,2\}$, with $\mathbb P\{U=1\}=p$: $\boldsymbol X=\boldsymbol X_U$.  Then, we have
\begin{align}
J_m(\boldsymbol X,\boldsymbol Z) \geq & J_m(\boldsymbol X, \boldsymbol Z|U)\\
 =& pJ_m(\boldsymbol X_1, \boldsymbol Z)+(1-p)J_m(\boldsymbol X_2, \boldsymbol Z)
\end{align} 
 which proves the concavity of $J_m(\boldsymbol X, \boldsymbol Z)$ for fixed $f_Z$. Similar arguments on $\boldsymbol Z$ prove that $J_m(\boldsymbol X, \boldsymbol Z)$ is also concave in $f_Z$ for fixed $f_X$.
\end{proof}


\subsection{Convexity of $J_r$ in $f_Y(\cdot)$ }
In this section, we assume that the source and the channel are scalar, i.e., $m\!=\!k\!=\!1$, for simplicity, although our results can  be extended to higher matched dimensions, $m=k, \forall m,k \in \mathbb N$. We show the convexity of $J(X, Z, g, h)$ in the channel input density $f_Y(\cdot)$ of $Y=g(X)$, when $h(\cdot)$ is optimized for $g(\cdot)$. An important distinction to make  is that convexity in  $g(\cdot)$ is not implied. A trivial example to demonstrate  non-convexity in $g(\cdot)$ is the scalar Gaussian source and channel setting, where both $Y=\sqrt {\frac{P_T}{\sigma_X^2}}  X$ and  $Y=-\sqrt {\frac{P_T}{\sigma_X^2}}  X$ are optimal (when used in conjunction with their respective optimal decoders). This example also leads to the intuition that the cost functional may be ``essentially" convex (i.e., convex  up to the sign of $g(\cdot)$) although it is clearly not convex in the strict sense. It turns out that this intuition is correct:  $ J_r(X, Z, g)$ is convex in  $f_Y(\cdot)$.  

Towards showing convexity, we first introduce the idea of probabilistic (random) mappings, similar, in spirit, to the random encoders used in the coding theorems\cite{shannon1949mathematical, csiszar2011information}. We reformulate the mapping problem by allowing random mappings, i.e., we relax the mapping from a deterministic function $Y=g(X)$ to a probabilistic transformation, expressed as $f_{Y|X}(x,y)$. Note that similar relaxation to stochastic settings have been used in the literature, e.g. recently in \cite{wu2011witsenhausen}. We define this ``generalized" mapping problem as: minimize $J_{gen}(X, Y, Z)$  over the conditional density ${f_{Y|X}}$ where the cost functional $J_{gen}$ is defined as 
\begin{equation}
J_{gen}(X, Y, Z)\triangleq \inf \limits_{h}\mathbb E\{  (X-h(Y+Z))^2\} +\lambda \mathbb E\{ Y^2\}.
\label{gcost}
\end{equation}
We first need to show that this relaxation does not change the solution space. This is done via the following lemma. 

\begin{lemma}
$Y \sim f_Y(\cdot)$ which minimizes (\ref{gcost}) is a deterministic function of the input, $Y=g(X)$, i.e., $J_m(X,Z)=\inf \limits_{g} J_r(X,Z,g)=\inf \limits_{f_{Y|X}}J_{gen}(X, Y, Z)$.
\label{useful3}
\end{lemma}
\begin{proof}
Let us first define an auxiliary function 
\begin{equation}
G(X,Y,Z)\triangleq   (X-h(Y+Z))^2+Y^2.
\end{equation}
Next, we observe that 
\begin{IEEEeqnarray}{rCl}
 \inf_{h}\inf \limits_{f_{Y|X}}J_{gen}(X, Y, Z)
 &= \inf_{h}& \int \int\! f_X(x) \inf \limits_{f_{Y|X}} \left \{  \int G(X,Y,Z) f_{Y|X}(x,y) dy  \!  \right\} \!dx f_Z(z) dz  \\
& = \inf_{h} &\int \int\! f_X(x) \inf_{y} G(x,y,z) dxf_Z(z) dz, 
 \label{ken}
\end{IEEEeqnarray}
where (\ref{ken}) is due to the fact that the minimizing $f_{Y|X}$ simply allocates all probability to the value $y$ which minimize $G(x,y,z)$. Hence, for any fixed $h(\cdot)$, the minimizing $f_{Y|X}$ is deterministic.  Using the optimal $h(\cdot)$ as the fixed $h(\cdot)$ in (\ref{ken}), we show that  the optimal $Y \sim f_Y(\cdot)$  is a deterministic function: $Y=g(X)$. 

\end{proof}

Next, we proceed to show that the generalized mapping problem is convex in $f_Y(\cdot)$. To this aim, we show that $J_{gen}$ can be written in terms of a known metric in probability theory, Wasserstein metric \cite{villani2009optimal} and use its functional properties. First, we present the definition and some important properties of this metric. 

Wasserstein metric is a metric defined on the quadratic Wasserstein space\footnote{The quadratic Wasserstein space on $\mathbb R$ is defined as the collection of all Borel probability measures with finite second moments, denoted by ${\cal P}_2(\mathbb R)$.} ${\cal P}_2(\mathbb R)$, defined for $S, Q \in {\cal P}_2 (\mathbb R)$ as 
\begin{equation}
W_2(S, Q)=\inf \left \{ ||X-Y||_2 : X\sim S, Y\sim Q\right \},
\end{equation}
where $ ||X-Y||_2\triangleq\sqrt{\mathbb E\{(X-Y)^2\}}$ and the infimum is over the joint distribution of $X$ and $Y$.

The $W_2$ metric measures the convergence in distribution and second order moments, i.e., $W_2(S_{X_k}, S_X)$ converges to zero if and only if $X_k$ converges to $X$ in distribution and $\mathbb E \{X_k^2\}$ converges to $\mathbb E \{X^2\}$. The following properties of this  metric will be used to derive the  subsequent results.

\begin{lemma}[\cite{villani2009optimal}]
\label{useful2}
$W_2(S, Q)$ satisfies the following properties: \\
1) The metric $W_2(S, Q)$ is lower semi-continuous in both $S$ and $Q$. \\
2) For a given  $S$, $ W_2^2(S,Q)$  is convex in  $Q$. 
\end{lemma}

 Next, we present our main result in this section. Hereafter, we limit the space of decoding functions to ${\cal S}^+$, i.e., monotone increasing, without any loss of generality (see e.g., \cite{wu2012functional}).

\begin{theorem}
 $J_r$ is convex in $f_Y(\cdot)$ and hence the solution to the mapping problem is unique in $f_Y(\cdot)$.
\label{theo4}
\end{theorem}

\begin{proof}
We will first express $J_{m} (X,Y)$ as a minimization over $f_Y(\cdot)$. Let us define $V=h(Y+Z)$ for a fixed $h(\cdot)$. Then, using Lemma \ref{useful3}, $J_{m} (X,Y)$ can be re-written as 
\begin{equation}
J_{m} (X,Y)=  \inf \limits_{h} \inf \limits_{f_{V|X}} \mathbb E\{(X-V)^2\}+\inf \limits_{f_{Y|X}} \lambda \mathbb E\{Y^2\}
\end{equation}
which is 
\begin{equation}
J_{m} (X,Y)=  \inf \limits_{h} \inf \limits_{f_V}  \left \{W_2^2(f_X,f_V)\right \}+\lambda \inf \limits_{f_{Y}} \mathbb E\{Y^2\} .
\label{upeq1}
\end{equation}

The first term in the right hand side of (\ref{upeq1}) is convex in $f_V$ since, $W_2^2(f_X,f_V)$ is convex in $f_V$ (due to Lemma \ref{useful2}-property 2) when $f_X$ is fixed, and the pointwise minimizer of a convex function is convex. Since $Y$ and $V$ are related  in one-to-one manner through $V=h(Y+Z)$ and $h(\cdot) \in {\cal {S}}^+$ , this term is convex also in $f_Y$. Since  $\mathbb E\{Y^2\} $ is linear in $f_Y(\cdot)$, we conclude that $J_{m}$ is the infimum of a convex functional of $f_Y(\cdot)$, where the infimum is taken over  $f_Y(\cdot)$. Hence, the solution is unique in $f_Y(\cdot)$. 

Given that the solution is unique, we can express $J_m$ as 
\begin{equation}
J_m(X,Z)= \inf \limits_{f_Y} J_r (X,Z, g)
\end{equation}
{\it a.e.} in $X$ and $Z$,  where $Y=g(X)$. Hence the functional we are interested is indeed  $J_r (X,Z, g)$ which is convex $f_Y(\cdot)$.

\end{proof}

A practically important consequence of Theorem \ref{theo4} is stated in the following corollary. 

\begin{cor}
$J_r$ is convex in $g(\cdot)$ where $ g(\cdot) \in {\cal S^+}$ (i.e., $g(\cdot)$ is monotone increasing).
\end{cor}

\begin{proof}
There is one-to-one mapping between $Y$ and the encoder $g(\cdot)$ as ${\cal F}_X (X)={\cal F}_Y (g(X))$ where ${\cal F}_X $ and ${\cal F}_Y$ denote the cumulative distribution functions of $X$ and $Y$ respectively.It follows from Theorem \ref{theo4} that for any $f_{Y_1}$ and $f_{Y_2}$ and $1\geq \alpha \geq 0$
\begin{equation}
\alpha  J_r(f_{Y_1})+(1-\alpha )  J_r(f_{Y_2})  \geq J_r(\alpha f_{Y_1}+(1-\alpha )f_{Y_2}).
\end{equation}
Since $J_r(f_{Y})$ is achieved by a unique $g(\cdot) \in \cal G^+$, this implies that 
\begin{equation}
\alpha  J_r(f_{g_1})+(1-\alpha )  J_r(f_{g_2})  \geq J_r(\alpha f_{g_1}+(1-\alpha )f_{g_2}),
\end{equation}
which shows the convexity of $J_r$ in $g(\cdot)$, where $ g(\cdot) \in \cal G^+$. 
\end{proof} 

\begin{remark}
Note that the optimal mappings, i.e., the mappings that achieve the infimum in (\ref{gcost}) exist. To see this, we use the semi-lower continuity property of the $W_2(S, Q)$ in both $S$ and $Q$ as given in Lemma \ref{useful}. The set of $Y$ is compact since $\mathbb E \{Y^2\} \leq P_T$, hence the infimum in the problem definition is achievable. This result guarantees that any algorithm based on the necessary conditions of optimality will converge to a globally optimal solution. 
\end{remark}

\subsection{Optimality Conditions}

We proceed to develop the necessary conditions for optimality of the encoder and decoder subject to the average power constraint (\ref{power_cons}), in the general setting of $m, k \in \mathbb N$.
\subsubsection {Optimal Decoder Given Encoder}
Let $\boldsymbol g(\cdot)$ be fixed. Then the optimal decoder is the MMSE estimator of $\boldsymbol X$ given $\boldsymbol {\hat Y}=\boldsymbol {\hat y}$, i.e.,
\begin{equation}
 \boldsymbol h(\boldsymbol {\hat y})= \mathbb E \{\boldsymbol X|\boldsymbol {\hat y}\}
\end{equation}
Plugging the expressions for expectation, we obtain   
\begin{equation}
   {\boldsymbol h(\boldsymbol {\hat y})}= \int  \, {\boldsymbol x}  \, \normalfont f_{X|\hat Y} (\boldsymbol x, \boldsymbol {\hat y})  { \, \mathrm{d}\boldsymbol x}     .
\end{equation}
 Applying Bayes' rule
 \begin{equation}
 f_{X|\hat Y}({\boldsymbol x, \boldsymbol {\hat y}})= \dfrac{f_X({\boldsymbol x}) f_{\hat Y|X}(\boldsymbol x, \boldsymbol {\hat y}) }{  \int  \, f_X({\bf x})  \,  f_{\hat Y|X}(\boldsymbol x, \boldsymbol {\hat y})\,  \mathrm{d}\boldsymbol x}  
 \end{equation}
and noting that $f_{\hat Y|X}(\boldsymbol x, \boldsymbol {\hat y})= f_{Z}[{\boldsymbol {\hat y}-\boldsymbol g(\boldsymbol x)}]$, the optimal decoder can be written, in terms of known quantities, as        
\begin{equation} 
\label{decoder_map}
\boldsymbol h(\boldsymbol{\hat y})= \dfrac{ \int {\boldsymbol x}  \,  f_X( {\boldsymbol x})  \,  f_{Z} [{\boldsymbol {\hat y}-\boldsymbol g( \boldsymbol x)}]\, { \mathrm{d}\boldsymbol x} } {  \int\,     f_X({\boldsymbol x})  \,  f_{Z} [{\boldsymbol {\hat y}-\boldsymbol g(\boldsymbol x)}]  \mathrm{d}\boldsymbol x}    . 
\end{equation}  

\subsubsection {Optimal Encoder Given Decoder}

Let $\boldsymbol h (\cdot)$ be fixed. To obtain necessary conditions we apply the standard method in variational calculus \cite{Luenberger}:
\begin{equation}
\label{abeq}
\frac{\partial}{\partial \epsilon} \bigg|_{\epsilon =0} J \left [{\boldsymbol g}({\boldsymbol x})+\epsilon {\boldsymbol  \eta} ({\boldsymbol x} ) ,\boldsymbol h \right ] =0,
\end {equation}
i.e., we perturb the cost functional for all admissible\footnote{
Our admissibility definition does not need to be very restrictive since it is used to derive a necessary condition. Hence, the only condition required for the admissible functions is  to be (Borel) measurable, that the integrals
exist, and that we can change the order of integration and differentiation.} variation functions $\boldsymbol \eta({\boldsymbol x})$. Since the power constraint is accounted for in the cost function, the variation function $\boldsymbol \eta(\cdot)$ needs not be restricted to satisfy the power constraint (all measurable functions ${\boldsymbol \eta} : \mathbb R^m\rightarrow \mathbb R^k$ are admissible). Applying (\ref{abeq}), we get
\begin{align}
\int &\left \{\lambda {\boldsymbol g(\boldsymbol x)}-\int {\boldsymbol h'(\boldsymbol g(\boldsymbol x)+\boldsymbol z)}[ {\boldsymbol x}-{\boldsymbol h (\boldsymbol g(\boldsymbol x)+\boldsymbol z)}] f_Z({\boldsymbol z})  \mathrm{d}\boldsymbol z \right \}{ \boldsymbol \eta({\boldsymbol x} ) } f_X({\boldsymbol x}) {\mathrm{d}\boldsymbol x}=0
\end{align}
\noindent where $\boldsymbol h'(\cdot) $ denotes the Jacobian of the vector valued function ${\boldsymbol h (\cdot)}$. Equality for all admissible variation functions, $\boldsymbol \eta(\cdot)$, requires the expression in braces to be identically zero (more formally the functional derivative \cite{Luenberger} vanishes at an extremum point of the functional). This gives the necessary condition for optimality as
\begin{equation}
\label {encoder_map} \nabla_{\boldsymbol g} J[{\boldsymbol g,\boldsymbol h}]=0, 
\end{equation}
where 
\begin{align}
 \nabla_{\boldsymbol g} J[{\boldsymbol g,\boldsymbol h}] =  \lambda  f_X({\boldsymbol x}) {\boldsymbol g(\boldsymbol x)}- f_X({\boldsymbol x})   \int   \boldsymbol  h'(\boldsymbol g(\boldsymbol x)\!+\boldsymbol z)  [{\boldsymbol x\!-\!\boldsymbol h(\boldsymbol g(\boldsymbol x)\!+\! \boldsymbol z)}]  f_Z({\boldsymbol z}) {\mathrm{d} {\boldsymbol z}} 
 \end{align}
 
\begin{remark} Unlike the decoder, the optimal encoder is not in closed form. 
\end{remark}
  We summarize the results in this section, in the following theorem.
  
   \begin {theorem}
 \label{pptheo}
 Given source and noise densities, a coding scheme ($\boldsymbol g(\cdot), \boldsymbol h(\cdot)$) is optimal {\it only if}
\begin{equation}
\label{enc}
{\boldsymbol g(\boldsymbol x)}  \! =\!  \frac{1}{ \!\lambda } \int { {\boldsymbol h'(g(\boldsymbol x)\!+\!z)  [{\boldsymbol  x\!-\!\boldsymbol h(\boldsymbol g(\boldsymbol x)\!+\!\boldsymbol z)}]}    f_Z({\boldsymbol z})   {\mathrm{d} {\boldsymbol  z}} }
\end{equation}        
         \begin{equation}
          {\boldsymbol h(\boldsymbol {\hat y)}}= \dfrac{{  \int} {\boldsymbol x}  \,  f_X( {\boldsymbol x})  \,  f_{Z} [{\boldsymbol {\hat y}-\boldsymbol g(\boldsymbol x)}]\, \mathrm{d} {\boldsymbol x} } {  \int  \,     f_X({\boldsymbol x})  \,  f_{Z} \left [{\boldsymbol {\hat y}-\boldsymbol g(\boldsymbol x)}\right ]  \, \mathrm{d} {\boldsymbol x}}
  \label{iki} 
    \end{equation}   
          where varying $\lambda$ provides solutions at different levels of power constraint $P_T$. In fact,  $\lambda$ is the slope of the distortion-power curve: $\lambda= -\frac{d D}{d P_T }$. Moreover, if $m=k=1$, these conditions are also sufficient for optimality. 
 \end {theorem} 
\begin{proof}
The necessary conditions were derived in (28), (31) and (32). The sufficiency in the case of $m=k=1$ follows from Corollary 1.
\end{proof} 

The theorem states the necessary conditions for optimality but they are not sufficient in general, as is demonstrated in particular by the following corollary. 

\begin{cor}
For Gaussian source and channel, the necessary conditions of Theorem \ref{pptheo} are satisfied by linear mappings $\boldsymbol g(\boldsymbol x)= K_e\boldsymbol x$ and $\boldsymbol h(\boldsymbol y)= K_d\boldsymbol y$  for some $K_e \in \mathbb R^{m\times k},   K_d \in \mathbb R^{k\times m}$ for any $m,k \in \mathbb N$. 
\label{maincor}
\end{cor} 

\begin{proof}
Linear mappings satisfy the first necessary condition, (\ref{enc}), regardless of the source and channel densities. Optimal decoder is linear in the Gaussian source and channel setting, hence the linear encoder-decoder pair satisfies both of the necessary conditions of optimality. 
\end{proof}

Although linear mappings satisfy the necessary conditions of optimality for the Gaussian case, they are known to be highly suboptimal when dimensions of source and channel do not match, i.e., $m \neq k$, see e.g. \cite{basar1980performance}. Hence, this corollary illustrates the existence of poor local optima and the challenges facing algorithms based on these necessary conditions.

\subsubsection{Extension to distributed settings}
Optimality conditions for the setting of decoder side information can be obtained by following similar steps. However, they involve somewhat more complex expressions and are relegated to the appendix.  We note, in particular, that for these settings a similar result to Corollary \ref{maincor} holds, i.e., for Gaussian sources and channels linear mappings satisfy the necessary conditions. Perhaps surprisingly, even in the matched bandwidth case, e.g., scalar source, channel and side information, linear mappings are strictly suboptimal. This observation highlights the need for powerful numerical optimization tools.

\section{Algorithm Design}

The basic idea is to iteratively alternate between the imposition of individual necessary conditions for optimality, and thereby successively decrease the total Lagrangian cost. Iterations are performed until the algorithm reaches a stationary point. Imposing optimality condition for the decoder is straightforward, since the decoder can be expressed as  closed form functional of known quantities, ${\boldsymbol g} (\cdot)$, $ f_X(\cdot)$ and $f_Z(\cdot)$. The encoder optimality condition is not in closed form and we perform steepest descent search in the direction of the functional derivative of the Lagrangian with respect to the encoder mapping $\boldsymbol g(\cdot)$. By design, the Lagrangian cost decreases monotonically as the algorithm proceeds iteratively. The update for the various encoders is stated generically as
\begin{equation}
\label{g_iter} {\boldsymbol g}_{i+1}({\boldsymbol x})= {\boldsymbol g}_{i}({\boldsymbol x})-\mu \nabla_{\boldsymbol g}\boldsymbol J[\boldsymbol g, \boldsymbol h], 
\end{equation} 
where $i$ is the iteration index, $\nabla_{\boldsymbol g} \boldsymbol J[\boldsymbol g,{\boldsymbol h}]$ is the directional derivative, and $\mu$ is the step size.
At each iteration $i$, the total cost decreases monotonically and iterations are continued until convergence.  Previously proposed heuristic suboptimal mappings \cite{chung, ramstad} can be used as initialization for the encoder mapping optimization. Note that there is no guarantee that an iterative descent algorithms of this type will converge to the globally optimal solution. The algorithm will converge to a local minimum. An important observation is that, in the case of Gaussian sources and channels, the linear encoder-decoder pair satisfies the necessary conditions of optimality, although, as we will illustrate, there are other mappings that perform better. Hence, initial conditions have paramount importance in such greedy optimizations. A preliminary low complexity approach to mitigate the poor local minima problem, is to embed in the solution the noisy relaxation method of \cite{gadkari1999robust, Knagenhjelm}. We initialize the encoding mapping with random initial conditions and run the algorithm at very low CSNR (high Lagrangian parameter $\lambda$). Then, we gradually increase the CSNR (decrease $\lambda$) while tracking the minimum until we reach the prescribed CSNR (or power $P_T$ for a given channel noise level). The numerical results of this algorithm is presented in Section VI.

\section{On Linearity of Optimal Mappings}

 In this section, we address the problem of ``linearity" of optimal encoding and decoding mappings. Our approach builds on \cite{emrah_estimation}, where conditions for linearity of optimal estimation are derived, and the convexity result in Theorem 2 and the necessary conditions for optimality presented in Theorem 3. In this section, we focus on the scalar setting, $m=k=1$, while our results can  be extended to more general settings.   

\subsection{Gaussian Source and Channel}
We briefly revisit the special case in which both $X$ and $Z$ are Gaussian, $X\sim \mathcal N(0,\sigma_X^2)$ and $Z \sim \mathcal  N(0,\sigma_Z^2)$. It is well known that the optimal mappings are linear, i.e., $g(X)=k_eX$ and $h(Y)= k_dY$ where $k_e$ and $k_d$ are given by
\begin{equation}
k_e=\sqrt {\frac{ P_T}{\sigma_X^2}}, \quad  k_d=\frac{1}{k_e} \left (\frac{P_T}{P_T+\sigma_Z^2} \right ).
\label{Gauss}
\end{equation}


\subsection{On Simultaneous Linearity of  Optimal Encoder and Decoder}
In this section, we show that optimality requires that either both mappings be linear or that they both be nonlinear, i.e., a linear encoder with a nonlinear decoder, or a nonlinear encoder in conjunction with  a linear decoder, are both strictly suboptimal. We show this in two steps in the following lemmas. 

\begin{lemma}
\label{lemma1}
The optimal encoder is linear {\it a.e.} if the optimal decoder is linear. 
\end{lemma}

\begin{proof}
Let us plug $h(y)=k_dy$ for some $k_d \in \mathbb R$ in the necessary condition of optimality (\ref{enc}). Noting that $h'(y)=k_d$ {\it a.e.} in $y$, we have  
\begin{equation}
\lambda g(x)=k_d \int (x-k_d g(x)-k_dz) f_Z(z) dz 
\label{lin1}
\end{equation} 
{\it a.e.}  in $x$. Evaluating the integral and noting that $\mathbb E\{Z\}=0$, we have 
\begin{equation}
\label{linn2}
\lambda g(x)=k_d (x-k_d g(x))
\end{equation} 
{\it a.e.} and hence $g(x)=\frac{k_d}{\lambda+k_d^2}x \triangleq k_e x$.  \end{proof}
 
\begin{lemma}
The optimal decoder is linear {\it a.e.} if the optimal encoder is linear. 
\label{lemma2}
\end{lemma}
\begin{proof}
Plugging $g(x)=k_ex$ for some $k_e \in \mathbb R$ in the necessary condition of optimality (\ref{enc}), we obtain
\begin{equation}
\lambda k_e x=\int (x-h(k_e x+z)) h'(k_ex+z) f_Z(z) dz 
\label{lin2}
\end{equation} 
{\it a.e.} in $x$. Since $h(\cdot)$ is a function from $ \mathbb R  \rightarrow \mathbb R$, Weierstrass theorem \cite{dudley2002real} guarantees that there is a sequence of real valued polynomials that uniformly converges to it:
\begin{equation}
\label{lin13}
h(y)=\lim \limits_{i\rightarrow \infty}\sum \limits_{r=0}^{\infty} \alpha_r(i)y^r
\end{equation} 
where $\alpha_r(i)\in\mathbb R$  is the $r^{th}$ polynomial coefficient of the $i^{th}$ polynomial. Since Weierstrass convergence is uniform in $y$, we can interchange the limit and summation and hence, 
\begin{equation}
\label{lin3}
h(y)=\sum \limits_{r=0}^{\infty} \alpha_r y^r
\end{equation} 
{\it a.e.} in $y$, where $\alpha_r=\lim \limits_{i\rightarrow \infty}\alpha_r(i)$. Plugging (\ref{lin3}) in (\ref{lin2}) we obtain 
\begin{align}
\label{lin4}
\lambda k_e x={\int} \left (x-\sum \limits_{i=0}^{\infty} \alpha_i (k_e x+z)^i \right )\left (\sum \limits_{i=0}^{\infty}  i\alpha_i (k_e x+z)^{i-1} \right ) f_Z(z)   dz .
\end{align} 
Interchanging the summation and integration\footnote{Since the polynomials  $\sum \limits_{i=0}^{\infty} \alpha_i (k_e x+z)^i$ and $\sum \limits_{i=0}^{\infty} i \alpha_i (k_e x+z)^{i-1}$  respectively converge to $h(k_e x+z)$ and $h(k_e x+z) $ uniformly in $x$ and $z$, and hence both upper bounded in magnitude, we can use Lebesgue's dominated convergence theorem to interchange the summation and the integration.}
\begin{align}
\label{lin6}
-\lambda k_e x+ x-\sum \limits_{i=0}^{\infty} i \alpha_i \int (k_e x+z)^{i-1} f_Z(z)   dz   =\sum \limits_{i=0}^{\infty}  \sum  \limits_{j=0}^{\infty}  i\alpha_i \alpha_j \int (k_e x+z)^{i-1}  (k_e x+z)^{j}    f_Z(z)   dz .
\end{align} 
Note that the above equation must hold {\it a.e.} in $x$, hence the coefficients of $x^r$ must be identical for all $r \in \mathbb N$. Opening up the expressions $(k_e x+z)^{i-1}$ and $(k_e x+z)^{j} $ via binomial expansion, we have  the following set of equations 
\begin{align}
\label{lin5}
 \sum \limits_{i=r+1}^{\infty} i  \binom{i-1}{r} \alpha_i\mathbb E\{Z^{i-1-r}\}  =\sum  \limits_{i=0}^{\infty}\sum  \limits_{j=0}^{\infty}  \sum \limits_{l=0}^{i-1} \sum \limits_{p=r-l+1}^{j-1} \binom {j}{p} \binom{i-1}{l}  i\alpha_i \alpha_j \mathbb E\{Z^{i+j-1-p-l}\} 
\end{align} 
which must hold for all $r \geq 2$.

We note that every equation introduces a new variable $\alpha_r$,  so each new equation is linearly independent of its predecessors. Next, we solve these equations recursively, starting from $r=1$. At each $r$, we have one unknown ($\alpha_r$) which is related ``linearly" to known constants. Since the number of linearly independent equations is equal to the number of unknowns for each $r$, there must exist a unique solution. We know that $\alpha_r =0,$ for all $r \geq 2$  is a solution to (\ref{lin5}),  so it is the only solution.

\end{proof}

Next, we summarize our main result pertaining to the simultaneous linearity of optimal encoder and decoder.

\begin{theorem}
\label{simultenous}
The optimal mappings are either both linear or they are both nonlinear.
\end{theorem}
\begin{proof}
The proof directly follows from Lemma \ref{lemma1} and Lemma \ref{lemma2}.
\end{proof}

\subsection{Conditions for Linearity of Optimal Mappings}
In this section, we study the condition for linearity of optimal encoder and/or decoder. Towards obtaining our main result, we will use the following auxiliary lemma. 

\begin{lemma}
\label{useful}
The linear encoder and decoder in (\ref{Gauss}) satisfy the first of the necessary conditions of optimality (\ref{enc}) regardless of the source and channel densities. 
\end{lemma}
 \begin{proof}
 The proof directly follows from substitution of  (\ref{Gauss})  in (\ref{enc}).
 \end{proof}
 
 The following theorem presents the necessary and sufficient condition for linearity of optimal encoder and decoder mappings. 

\begin{theorem}
For a given power limit $P_T$, noise $Z$ with variance $\sigma_Z^2$ and characteristic function $F_Z(\omega)$, source $X$ with variance $\sigma_X^2$ and characteristic function $F_X(\omega)$, the optimal encoding and decoding mappings are linear {\it if and only if}
\begin{equation}
\label{matching}
F_X  (\alpha \omega)=F_Z^\gamma (\omega),
\end{equation}
where $\gamma =\frac{P_T}{\sigma_Z^2}$ and $\alpha=  \sqrt \frac{P_T}{\sigma_X^2}$.

\end{theorem}

\begin{proof}
Theorem \ref{simultenous} states that  the optimal encoder is linear if and only if optimal decoder is linear. Hence, we will only focus on the case where encoder and decoder are simultaneously linear. The first necessary condition is satisfied by Lemma \ref{useful}, hence  only the second necessary condition, (\ref{iki}) remains to be verified.

Plugging $g(X)=k_eX$ and $h({\hat Y})=k_d{\hat {Y}}$ in (\ref{iki}), we have
\begin{equation}
\label{upeq}
k_d{\hat y} =\dfrac{ { \int} {x}  \,  f_X( {x})  \,  f_{Z} ({{\hat y}-k_ex}) \mathrm{d} {x} } {  \int  \,     f_X({x})  \,  f_{Z} ({{\hat y}-k_ex})  \, \mathrm{d} {x}}.
\end{equation}
Expanding (\ref{upeq}), we obtain 
\begin{equation}
k_d{\hat y}   \int  \,     f_X({x})  \,  f_{Z} ({{\hat y}-k_ex})  \, \mathrm{d} {x}= {\int} {x}  \,  f_X( {x})  \,  f_{Z} \left({{\hat y}-k_ex}\right)\, \mathrm{d} {x}.
\end{equation}
Taking the Fourier transform of both sides and via change of variables $u\triangleq \hat y-k_ex $, we have 
\begin{align}
\int  \int  k_d(u+k_ex) f_X(x) f_Z(u) \exp(-j\omega (u+k_ex)) dx du =\int  \int  x f_X(x) f_Z(u) \exp(-j\omega (u+k_ex)) dx du
 \end{align}
and rearranging the terms, we obtain
 \begin{align}
 \left (\frac{1-k_ek_d}{k_ek_d} \right ) F_Z(\omega)F_X'(k_e\omega)=  F_X(k_e\omega)F_Z'(\omega).
 \end{align}
Noting that 
\begin{equation}
\gamma= \frac{P_T}{\sigma_Z^2}=\frac{k_ek_d}{1-k_ek_d},
\end{equation}
we have
  \begin{align}
 \frac{F_X'(k_e\omega)}{F_X(k_e \omega)}=\gamma \frac{  F_Z'(\omega)} {F_Z(\omega)}
 \end{align}
which implies
   \begin{align}
 (\log {F_X(k_e\omega)} )'= (\log  {F_Z ^{\gamma}(\omega)})'.
 \end{align}
The solution to this differential equation is
   \begin{align}
 \log {F_X (k_e\omega)}= \log { F_Z^{\gamma}(\omega)}+C
 \end{align}
 where $C$ is constant. Noting that $F_X (0)=F_Z (0)=1$, we determine $C=0$ and hence
  \begin{align}
 F_X(k_e \omega) =  F_Z ^{\gamma} (\omega).
 \label{cond}
 \end{align}
 Since the solution is essentially unique, due to Corollary 1, (\ref{cond}) is not only a necessary but also the sufficient condition for linearity of optimal mappings. 
\end{proof}

\subsection{Implications of the Matching Conditions}
In this section, we explore some special cases obtained by varying CNSR (i.e., $\gamma$)  and utilizing the matching conditions for linearity of optimal mappings. We start with a simple but perhaps surprising result.
\begin{theorem}
\label{match}
Given a source and noise of equal variance, identical to the power limit  ($\sigma_X^2=\sigma_Z^2=P_T$), the optimal mappings are linear {\em if and only if} the noise and source distributions are identical, i.e., $f_X(x)=f_Z(x), a.e.$ and in which case, the optimal encoder is $g(X)=X$ and the optimal decoder is $h(\hat Y)=\frac {1}{2}\hat Y$.
\end{theorem}

\begin{proof} It is straightforward to see from (\ref{matching}) that, at $\gamma=1$, the characteristic functions must be identical. Since the characteristic function uniquely determines the distribution \cite{billingsley2008probability}, $f_X(x)=f_Z(x)$, $ a.e.$.
\end{proof}

\begin{remark}
Note that Theorem \ref{match} holds irrespective of the source (and channel) density, which demonstrates the departure from the well known example of scalar Gaussian source and channel. 
\end{remark}

Next, we investigate the asymptotic behavior of optimal encoding and decoding functions at low and high CSNR. The results of our asymptotic analysis are of practical importance since they justify, under certain conditions, the use of linear mappings without recourse to complexity arguments at asymptotically high or low CSNR regimes. 

\begin{theorem}
In the limit $\gamma \rightarrow  0 $, the optimal encoding and decoding functions are asymptotically linear if the channel is Gaussian, regardless of the source. Similarly, as $\gamma \rightarrow \infty$, the  optimal mappings are asymptotically linear if the source is Gaussian, regardless of the channel.
\label{lintheo}
\end{theorem}

\begin{proof} The proof follows from applying the central limit theorem \cite{billingsley2008probability} to the matching condition (\ref{matching}). The central limit theorem states that as $\gamma \rightarrow  \infty $, for any finite variance noise $Z$, the characteristic function of the matching source $F_X(\omega)=F_Z^{\gamma}(\omega/k_e)$  converges to the Gaussian characteristic function. Hence, at asymptotically high CSNR, any noise distribution is matched by the Gaussian source. Similarly, as  $\gamma \rightarrow  0 $ and for any $F_X(k_e \omega)$, $F_X^{\frac{1}{\gamma} }(k_e \omega)$  converges  to the Gaussian characteristic function and hence the optimal mappings are asymptotically linear if the channel is Gaussian.
\end{proof}

Let us next consider a setup with given source and noise variables and a power which may be scaled to vary the CSNR, $\gamma$. Can the optimal mappings be linear at multiple values of $\gamma$ ? This question is motivated by the practical setting where $\gamma$ is not known in advance or may vary (e.g., in the design stage of a communication system). It is well-known that the Gaussian source-Gaussian noise pair makes the optimal mappings linear at all $\gamma$ levels. Below, we show that this is the only source-channel pair for which the optimal mappings are linear at more than one CSNR value.

\begin{theorem}
\label{gauss}
Given source and channel variables, let power $P_T$ be scaled to vary CSNR, $\gamma$. The optimal encoding and decoding mappings are linear at two different power levels $P_1$ and $P_2$ {\em if and only if} source and noise are both Gaussian. 
\end{theorem}

\begin{remark}
This theorem also holds for the setting  where  source or  noise variables are scaled to change CSNR for a given power $P_T$. 
\end{remark}

\begin{proof}
Let  $\gamma_1$ and $\gamma_2$ denote two CSNR levels, $g_1(X)=k_{e_1}X$ and $g_2(X)= k_{e_2}X$ denote encoding mappings. Let the power be scaled by $\alpha^2$ ($ \alpha  \in \mathbb R^+$), i.e., $P_2=\alpha^2 P_1$ which yields 
\begin{equation}
\label{snreq}
\gamma_2=\alpha^2 \gamma_1, \, \,
k_{e_2}=\alpha k_{e_1}.
\end{equation}
Using (\ref{matching}), we have
\begin{equation}
\label{simple}
F_{X}(k_{e_1} \omega)=F_{Z}^{\gamma_1}(\omega), F_{X}( k_{e_2} \omega)= F_{Z}^{\gamma_2}(\omega ).
\end{equation}
Hence,
\begin{equation}
\label{simple2}
F_{Z}^{\gamma_1}(\omega) = F_{Z}^{\gamma_2}( \alpha \omega).
\end{equation}
\noindent Taking the logarithm on both sides of (\ref{simple2}), applying (\ref{snreq}) and rearranging terms, we obtain
\begin{equation}
\label{eq_up}
\alpha^2=\frac{\log F_{Z}( \alpha \omega)}{\log F_{Z}(\omega)}.
\end{equation}
\noindent  Note that (\ref {eq_up}) should be satisfied for both $\alpha$ and $-\alpha$ since they yield the same $\gamma$.  Hence, $ F_{Z}( \alpha \omega)= F_{Z}( - \alpha \omega)$ for all $\alpha \in \mathbb R$, which implies $ F_{Z}(\omega)= F_{Z}( - \omega)$, {\it a.e.} in $\omega$. Using the fact that the characteristic function is conjugate symmetric (i.e., $F_{Z}(-\omega)=F_{Z}^*(\omega)$), we get $F_{Z}(\omega) \in \mathbb R,$ {\it a.e.} in $\omega$.
As $\log F_{Z}(\omega)$ is a function from $ \mathbb R  \rightarrow \mathbb C$, the Weierstrass theorem \cite{dudley2002real} guarantees that we can uniformly approximate $\log F_{Z}(\omega)$ arbitrarily closely by a polynomial 
 $ \sum \limits_{i=0}^{\infty}  k_i \omega^i $, where $k_i \in \mathbb C$. Hence, by (\ref{eq_up}) we obtain:
\begin{equation}
\label{above_eq_last}
\alpha^2=\frac{\sum \limits_{i=0}^{\infty}  k_i (\omega \alpha)^i}{\sum \limits_{i=0}^{\infty}  k_i \omega^i}
\end{equation}
{\it a.e.} in $\omega$ only if all coefficients $k_i$  vanish, except for $k_2$, i.e., $\log F_{Z}(\omega)=k_2\omega^2$, or $\log F_{Z}(\omega)=0$ {\it a.e.} in $\omega$ (the solution $\alpha=1$ is of no interest). The latter is not a characteristic function, and the former is the Gaussian characteristic function, $F_{Z}(\omega)=e^{k_2\omega^2}$, where we use the established fact that $F_{Z}(\omega) \in \mathbb R$.
Since a characteristic function determines the distribution uniquely, the Gaussian source and noise must be the only allowable pair.
\end{proof}

  \begin{figure*} 
  \centering
       \includegraphics[scale=0.60]{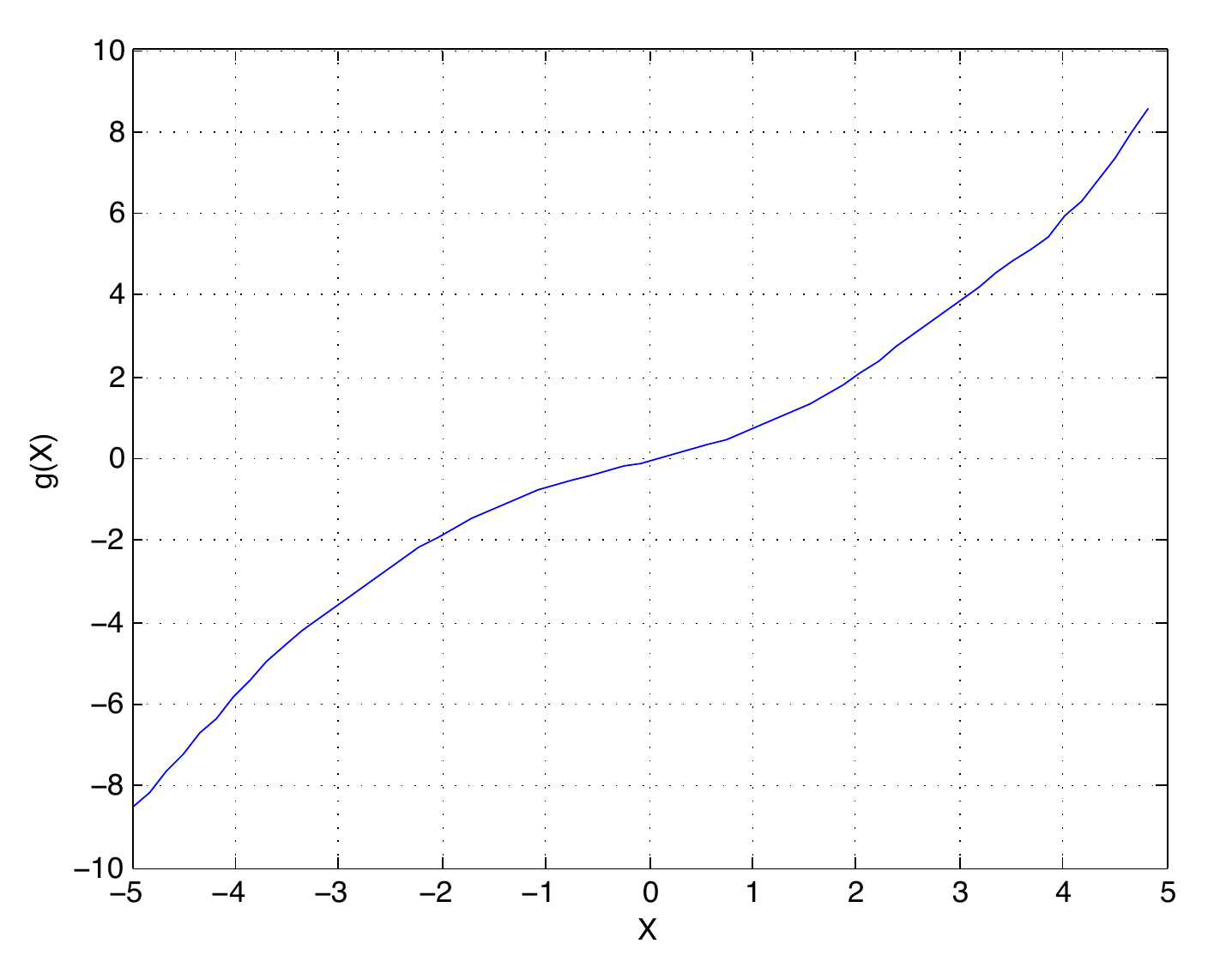}   
     \caption{Encoder mapping for bi-modal GMM source, Gaussian channel, modes at $3$ and $-3$ as in (\ref{density}). }
       \label{gmmtwo}
\end{figure*} 
\subsection{On the Existence of Matching Source and Channel}
Having discovered the necessary and sufficient condition as answer to the question of {\it when optimal zero-delay encoding and decoding mappings are linear}, we next focus on the question: {\it when can we find a matching source (or noise)  for a given noise (source)?}  Given  a valid characteristic function  $F_Z(\omega)$,  and for some $\gamma \in \mathbb R^+$, the function $F_Z^{\gamma}(\omega)$ may or may not be a valid characteristic function, which determines the existence of a matching source. For example, matching is guaranteed for integer $\gamma$ and it is also guaranteed for infinitely divisible $Z$. Conditions on $\gamma$ and  $F_Z(\omega)$ for $F_Z^{\gamma}(\omega)$ to be a valid characteristic function were studied detail in \cite{emrah_estimation}, to which we refer for brevity and to avoid repetition.

\section{Numerical Results}

We implement the algorithm described in Section IV by numerically calculating the derived integrals.  For that purpose,  we sample the source and noise distributions on a uniform grid. We also impose bounded support ($-5\sigma$ to  $+5\sigma$) i.e., neglected tails of infinite support distributions in the examples.

\subsection {Scalar Mappings ($m=1, k=1$), Gaussian Mixture Source  and Gaussian Channel}
We consider a Gaussian mixture source with distribution
\begin{equation}
\label{density}
 f_x(x)=\frac{1}{2\sqrt{2\pi}}\left \{e^{\frac{-(x-3)^2}{2}} +e^{\frac{-(x+3)^2}{2}}\right\} 
 \end{equation}
and unit variance Gaussian noise. The encoder and decoder mappings for this source-channel setting are given in Figure \ref{gmmtwo}. As intuitively expected, since the two modes of the Gaussian mixture are well separated, each mode locally behaves as Gaussian. Hence the curve can be approximated as piece-wise linear and deviates significantly from a truly linear mapping. This illustrates the importance of nonlinear mappings for general distributions that diverge from the pure Gaussian.

\subsection {A Numerical Comparison with Vector Quantizer Based Design }
In the following, we compare the proposed approach to the power constrained channel optimized vector quantization (PQCOVQ) based approach which first discretizes the problem, numerically solves the discrete problem and then linearly interpolates between the selected points. The main difference between our approach and PQCOVQ based approaches is that we derive the necessary conditions of optimality in the original, ``analog" domain without any discretization. This allows not only a theoretical analysis of the problem but also enables a different numerical method which iteratively imposes the optimality conditions of the ``original problem". However, PQCOVQ will arguably approximate the solution at asymptotically high sampling (discretization) resolution.  

  \begin{figure*}
\centering    
        \includegraphics[scale=0.5]{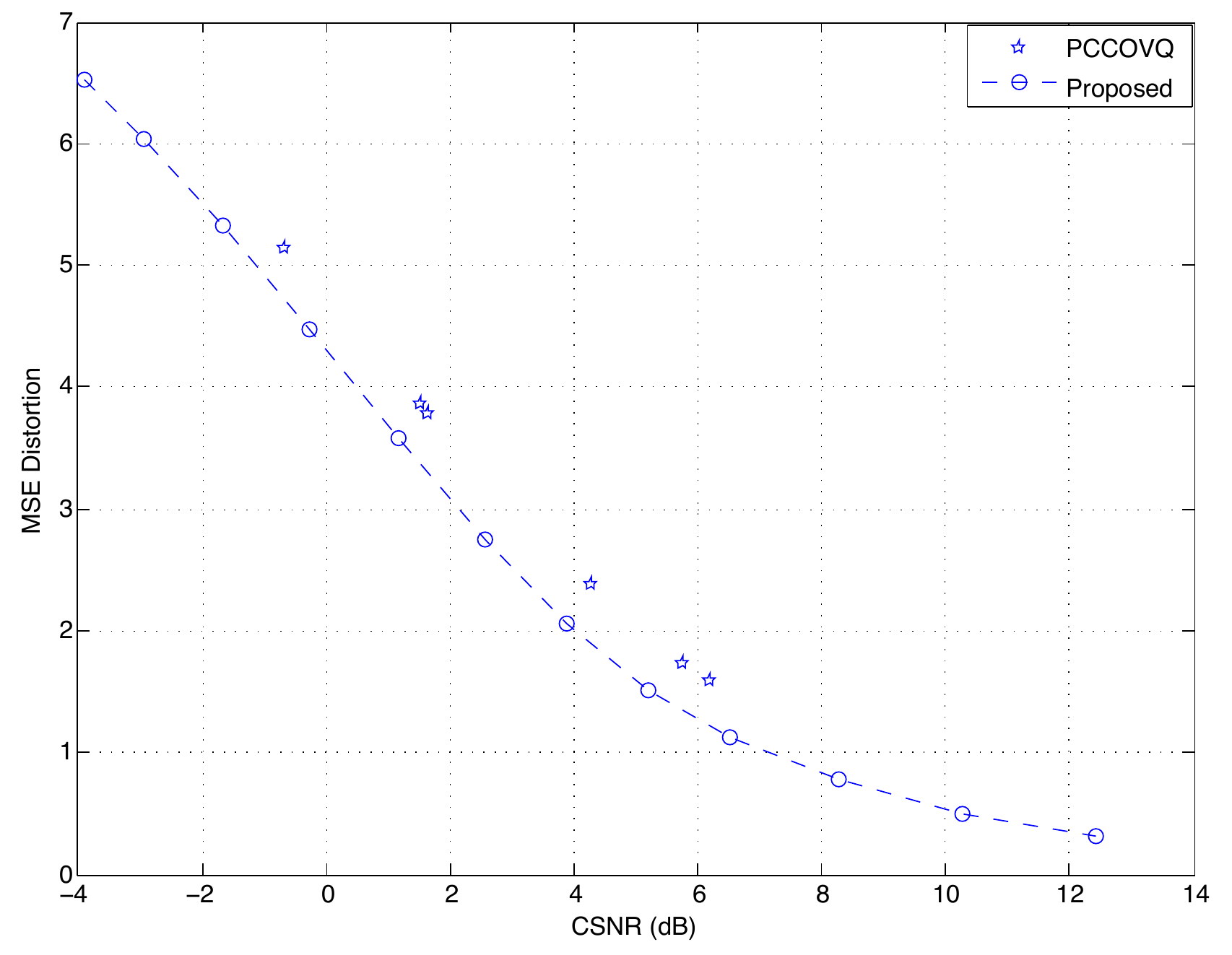}             
\caption{Comparative results: PQCOVQ vs. the proposed method.}
\label{pcc}
\end{figure*} 

 \begin{figure*}
\centering 
\subfigure[CSNR=-5.70dB]{
\includegraphics[scale=0.38]{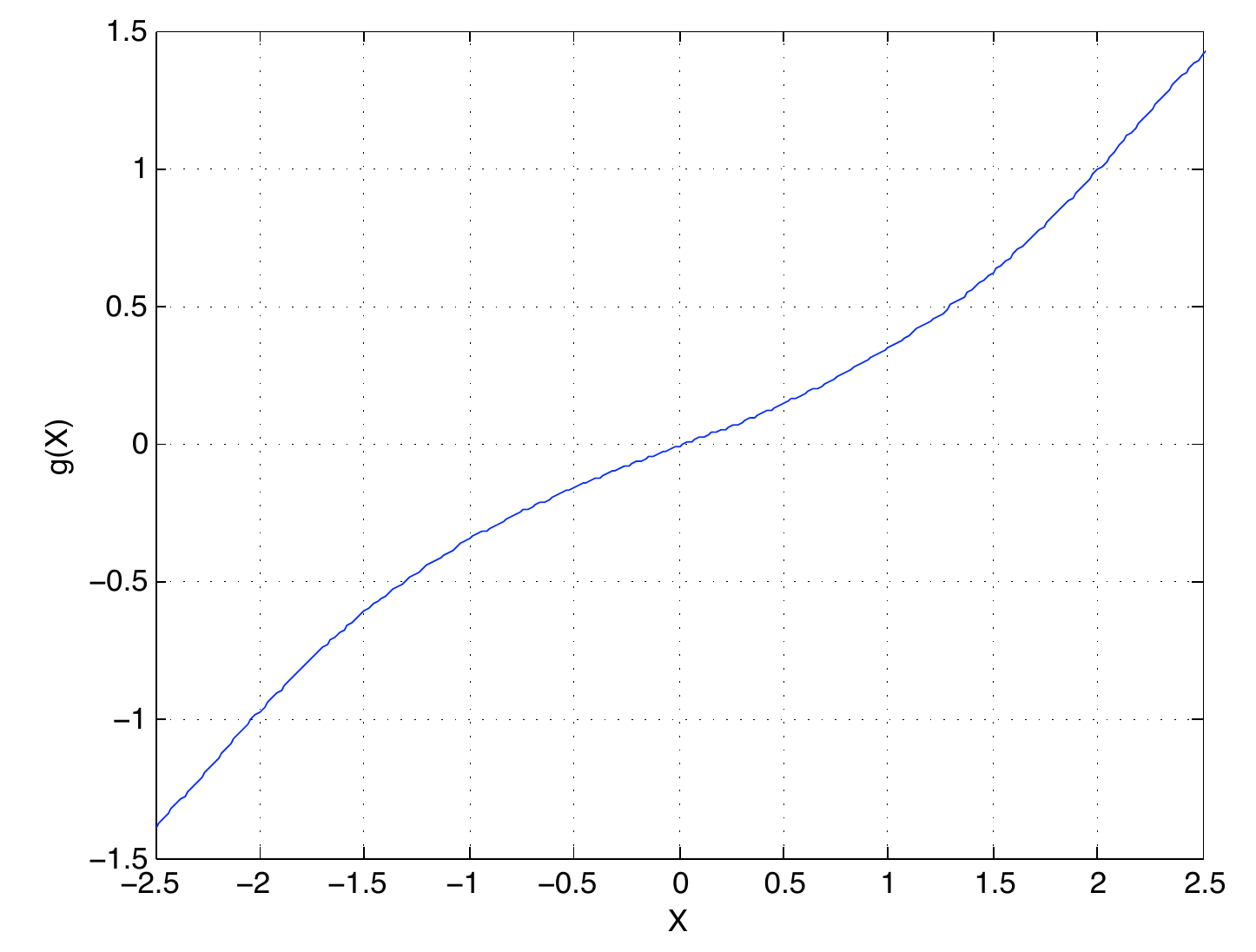}
}
\subfigure[CSNR=-1.69dB]{
\includegraphics[scale=0.38]{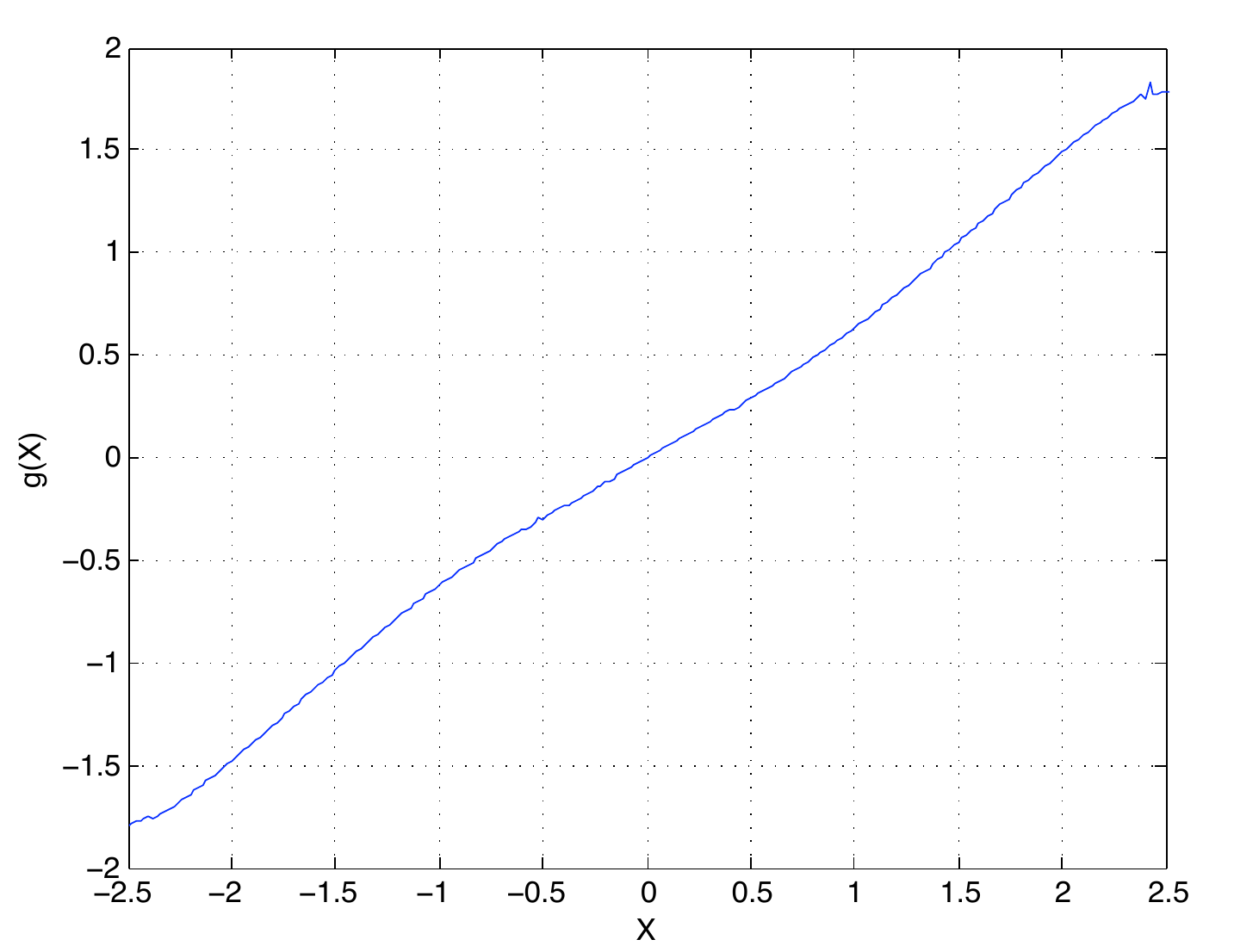}
}
\subfigure[CSNR=5.69dB]{
\includegraphics[scale=0.38]{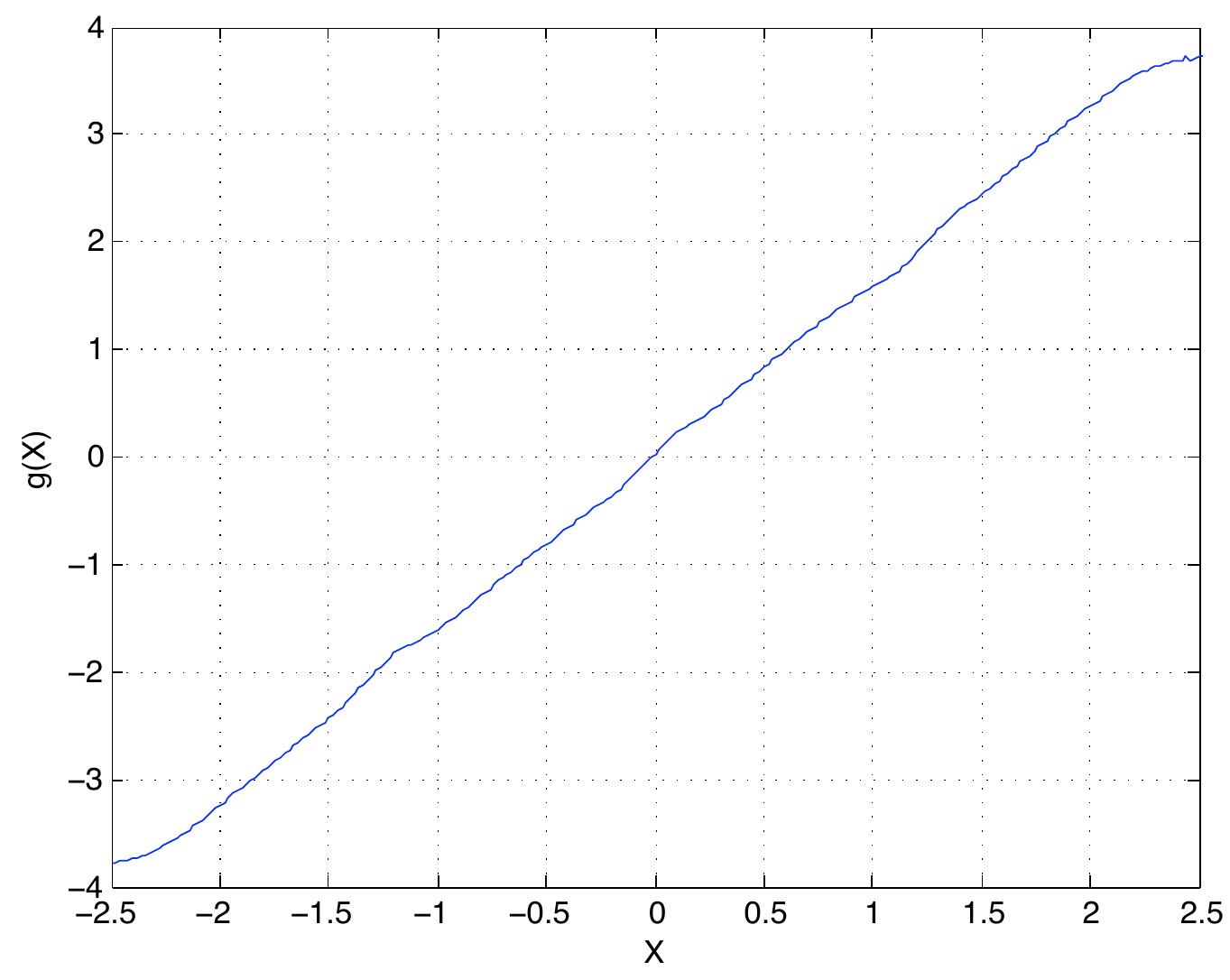}
}
\caption{This figure shows the optimal encoder at various CSNR values when $X \sim \mathcal N(0,1)$ and $Z$ is distributed uniformly on the interval $[-1,1]$ and CSNR is varied by changing power, $P$. Observe that the optimal encoder converges to linear as CSNR increases. }
\end{figure*}

\begin{figure*}
\centering 
\subfigure[CSNR=-5.70dB]{
\includegraphics[scale=0.38]{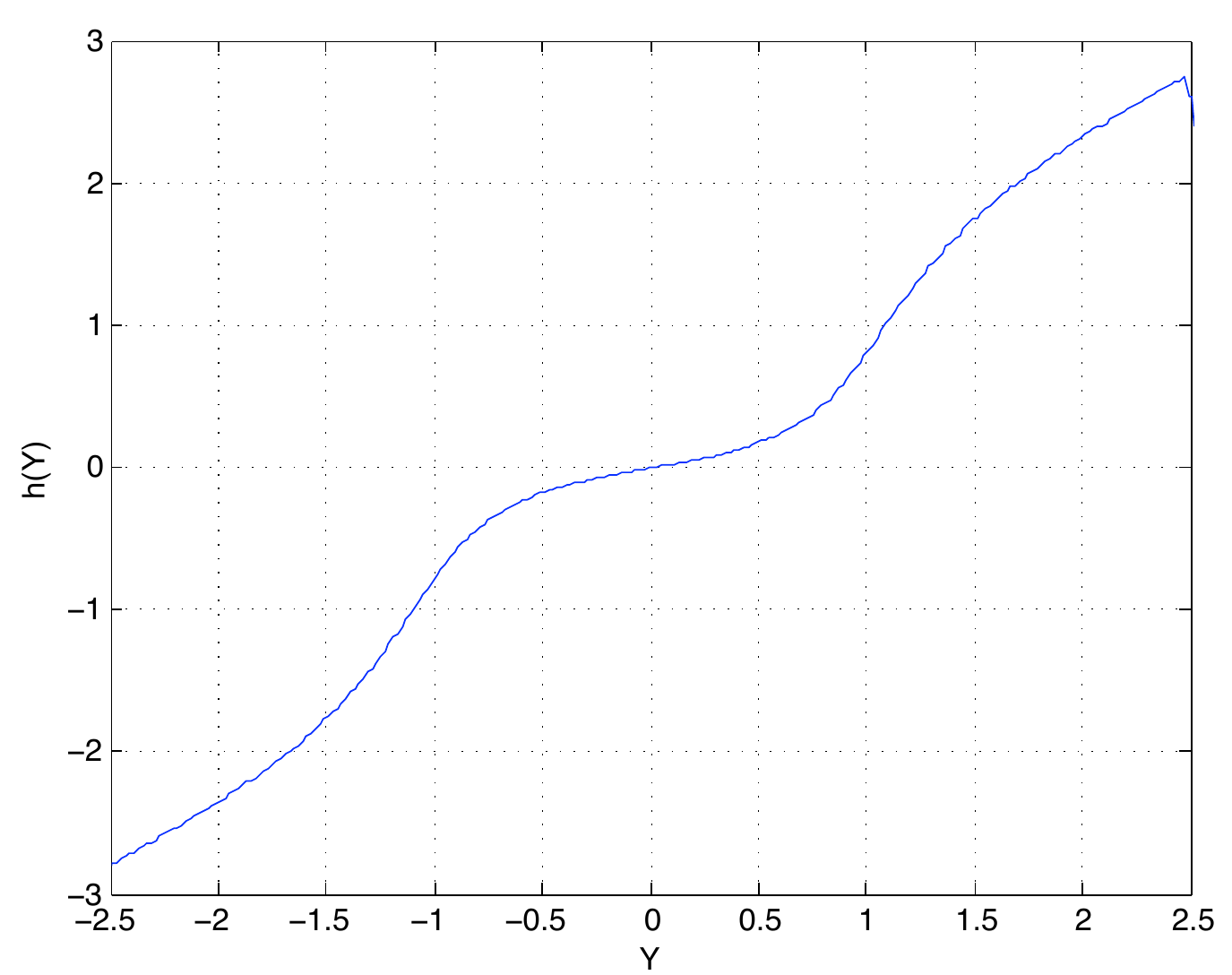}
}
\subfigure[CSNR=-1.69dB]{
\includegraphics[scale=0.38]{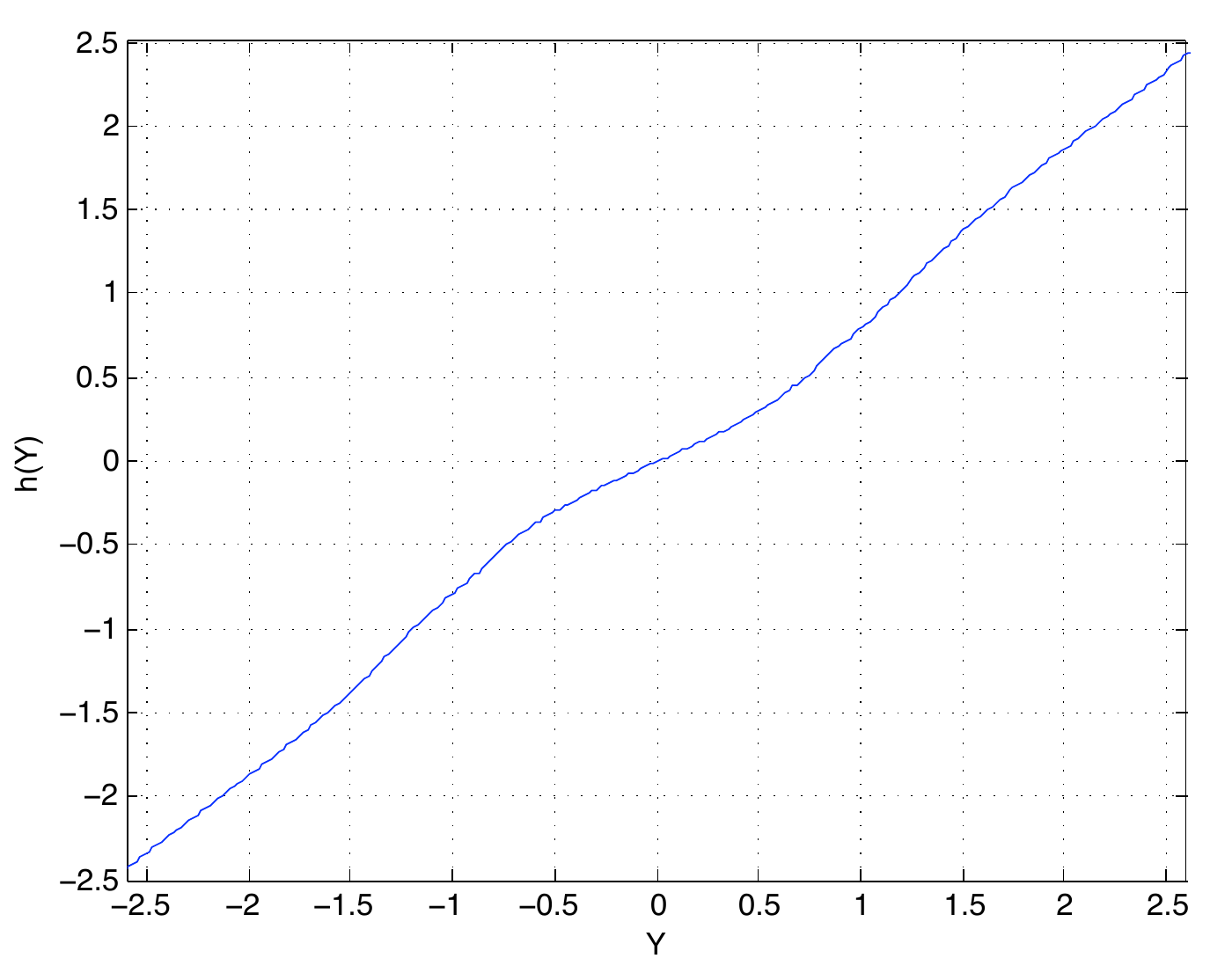}
}
\subfigure[CSNR=5.69dB]{
\includegraphics[scale=0.38]{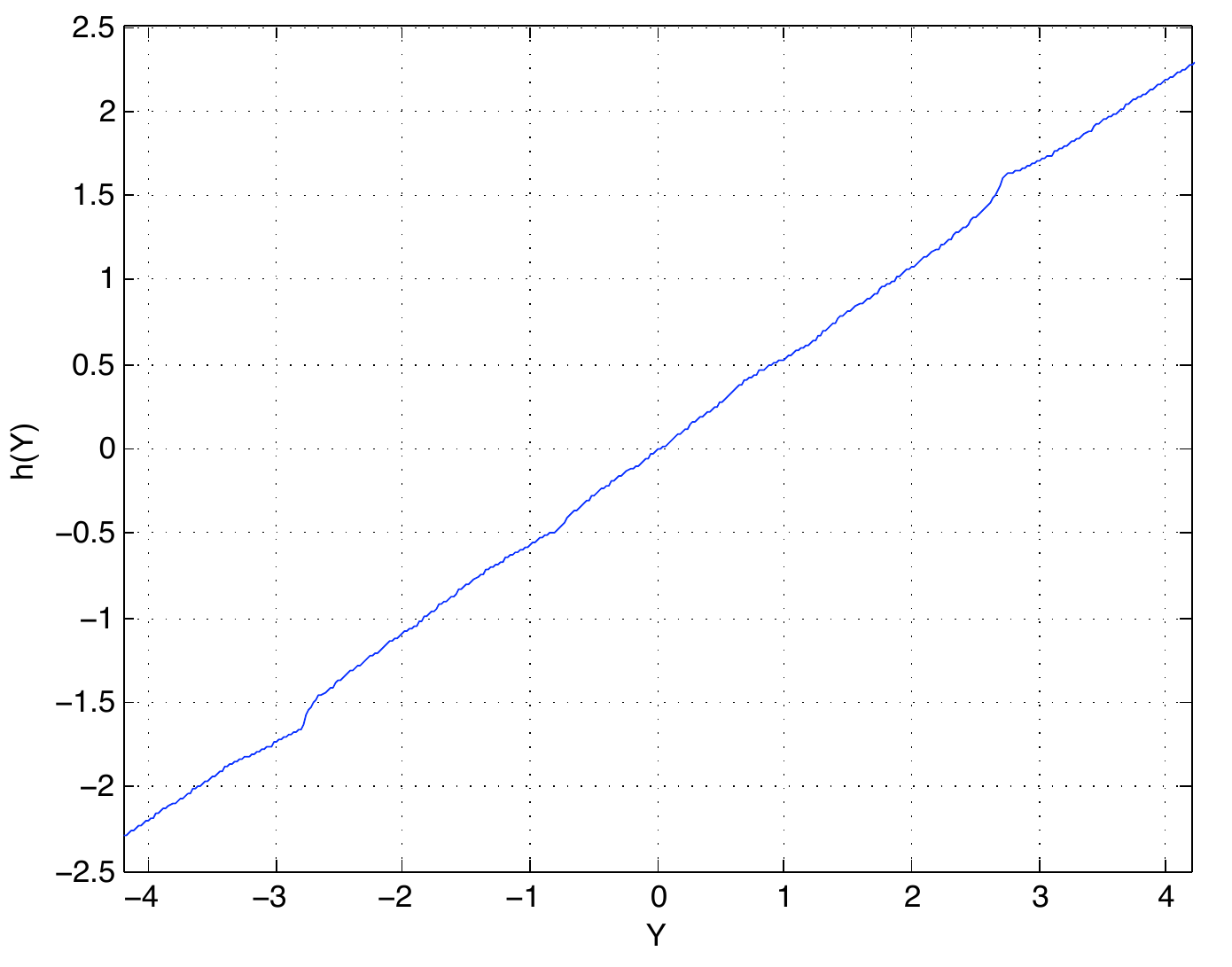}
}
\caption{This figure shows the optimal decoder (estimator) at various CSNR (in dB) values. Observe that the optimal decoder, similar to the optimal encoder in Figure 5, converges to linear as CSNR increases.}
\end{figure*}

To compare our method to PQCOVQ, we consider our running example of Gaussian mixture source and Gaussian channel. For both methods, we use 10 sampling points for the encoder mapping. The main difference is due to two facts: i) The proposed method is based on the necessary condition derived in the ``original" analog domain, and discretization is merely used to perform the ultimate numerical operations. On the other hand, PQCOVQ defines a ``discretized" version of the problem from the outset, with the implicit assumption that the discretized problem, at sufficiently high resolution, approximates well the original problem. Hence, although both methods eventually optimize and then interpolate a discrete set of points, the  proposed algorithm finds the values of these points while accounting for the fact that they will eventually be (linearly) interpolated. PQCOVQ does not account for eventual interpolation and merely solves the discrete problem. ii) Since we consider the problem in its original domain, we naturally use the optimal decoder, namely, conditional expectation. The PQCOVQ method uses the standard maximum likelihood method for decoding, see e.g.  \cite{karlsson2010optimized}.

The numerical comparisons are shown in Figure \ref{pcc}. As expected, the proposed method outperforms PQCOVQ for the entire range of CNSRs in this resolution constrained setting of 10 samples. We note  that  the performance difference diminishes at higher sampling resolution. The purpose of this comparison is to demonstrate the conceptual difference between these two approaches at finite resolution while acknowledging that the proposed method does not provide gains at asymptotically high resolution.

\begin{figure}\centering            \includegraphics[scale=0.60]
{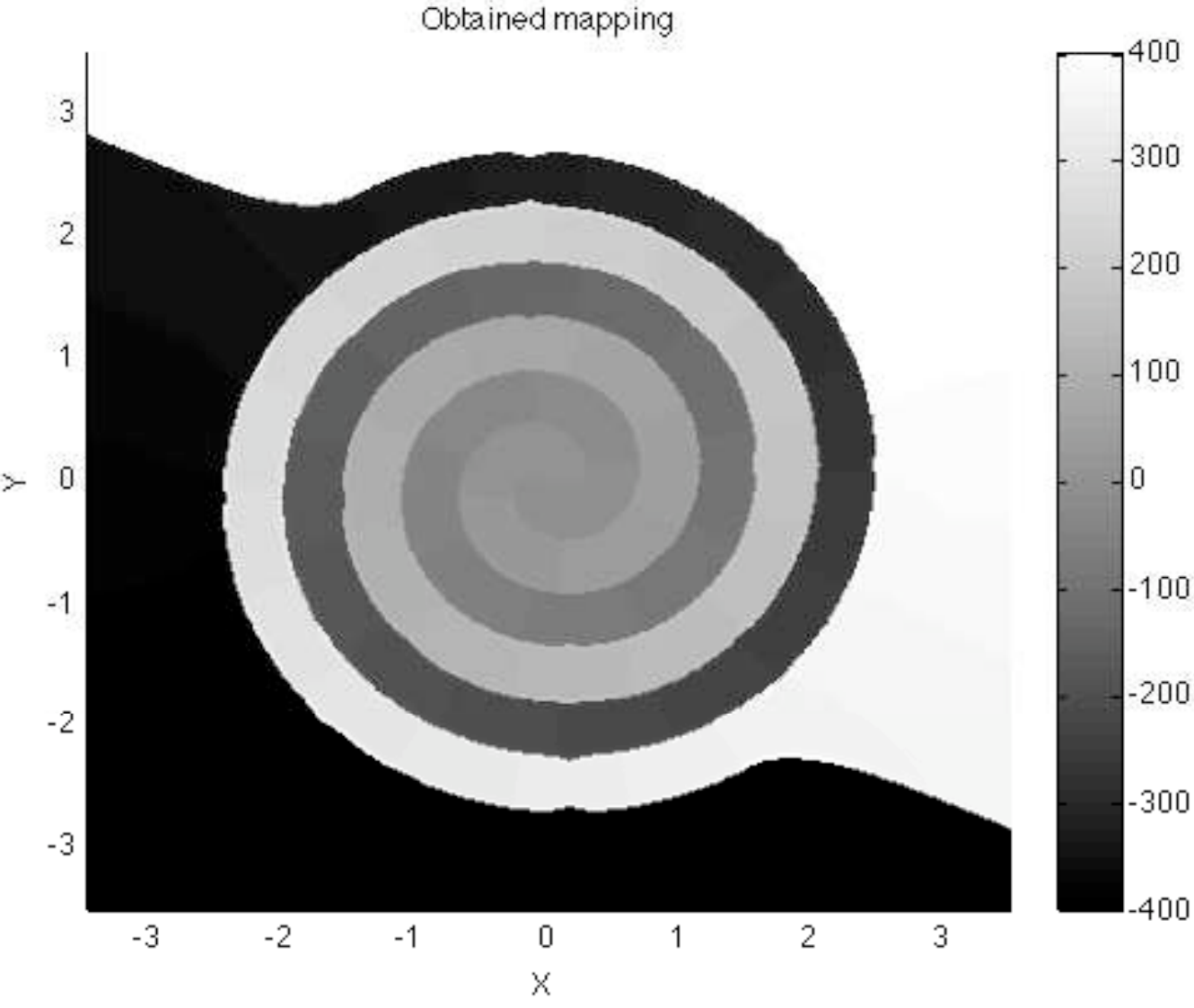}\caption{Encoder 2:1 mapping for unit variance Gaussian source and channel, at CSNR=40dB, SNR=19.41dB. The axes show the two dimensional input ($\boldsymbol x$) and the function value ($\boldsymbol g(\boldsymbol x)$) is reflected in the intensity level. } 
\label{fig21}
\end{figure}

\subsection {A Numerical Example for Theorem 7}


Let us consider a numerical example that illustrates the  findings in Theorem 7. Consider a setting where $X$ is Gaussian with unit variance and $Z$ is Gaussian with unit variance, i.e., $Z \sim \mathcal N(0,1)$. We change $\gamma$ (CSNR) by varying allowed power $P_T$, and observe how the optimal mappings behave for different $\gamma$.  We numerically calculated the optimal mappings by discretizing the integrals  on a uniform grid, with a step size $\Delta= 0.01$, i.e., to obtain the numerical results, we approximated the integrals as Riemann sums. Figures 5 and 6 respectively show how the optimal encoder and decoder mappings converge to linear as CSNR increases. Note that at $\gamma=-5.70$, optimal mappings are both highly nonlinear while at $\gamma=5.69$, they practically converge to linear, as theoretically anticipated from Theorem \ref{lintheo}.

\subsection{ ($m=2, k=1$)  Gaussian source-channel mapping}
In this section, we present a bandwidth compression example with 2:1 mappings for Gaussian vector source of size two (source samples are assumed to be independent and identically distributed with unit variance) and scalar Gaussian channel, to demonstrate the effectiveness of our algorithm in differing source and channel dimensions. We compare the proposed mapping to the asymptotic bound (OPTA) and prior work \cite{hekland}. We also compare the optimal encoder-decoder pair to the setting where only the decoder is optimized and the encoder is fixed. In prior work \cite{chung, ramstad, hekland}, the Archimedian spiral is found to perform well for Gaussian 2:1 mappings, 
and used for encoding and decoding with maximum likelihood criteria. We hence initialize our algorithm with an Archimedian spiral (for the encoder mapping). For details of the Archimedian spiral and its settings, see e.g.\cite{hekland} and references therein. 

\begin{figure}
\centering    
        \includegraphics[scale=0.65]{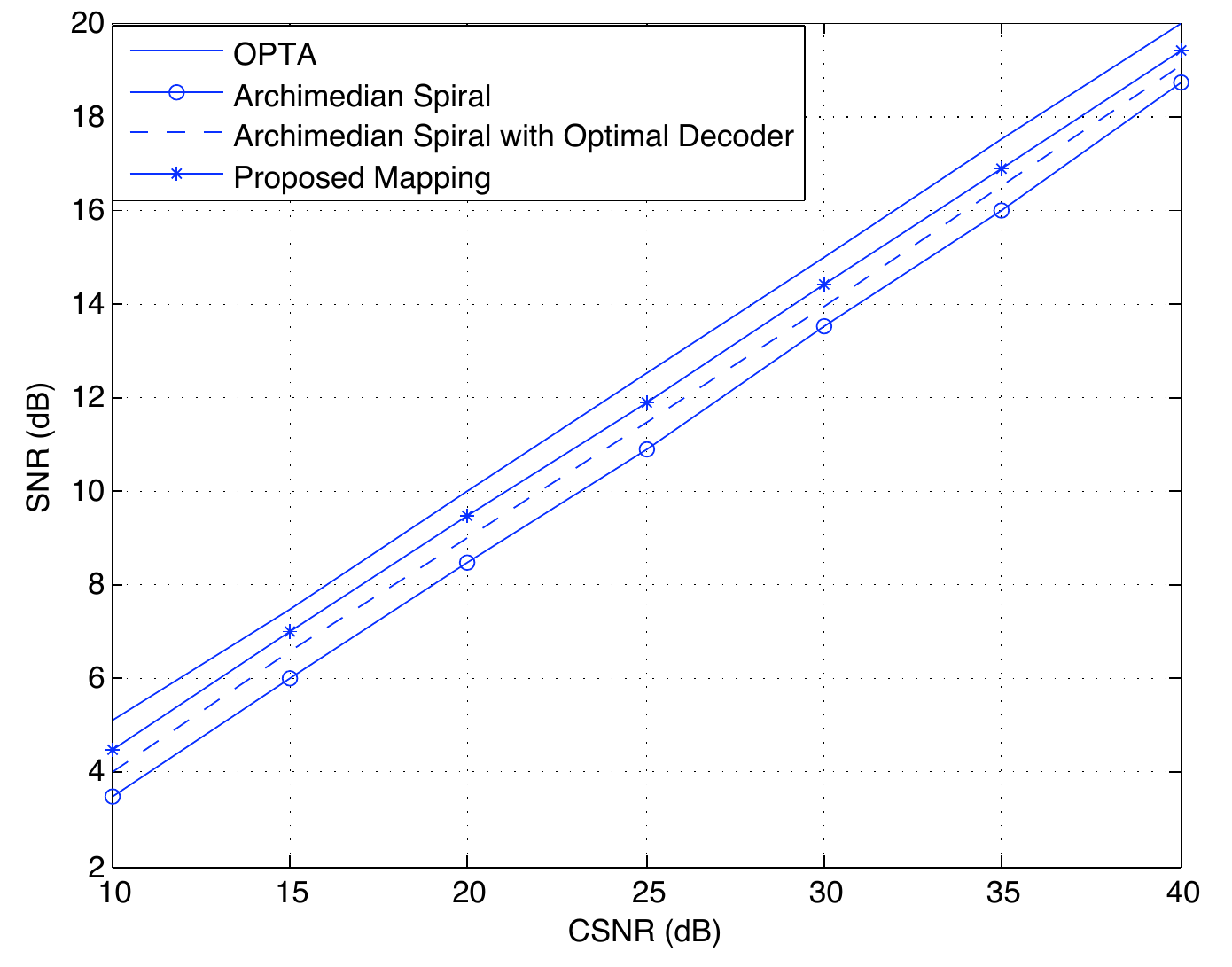}             
\caption{Comparative results for Gaussian source-channel, 2:1 mapping.}\label{res_21}
\end{figure}  

The obtained encoder mapping is shown in Figure \ref{fig21}. While the mapping produced by our algorithm resembles a spiral,  it nevertheless differs from the Archimedian spiral, as will also be evident from the performance results. Note further that the encoding scheme differs from prior work in that we continuously map the source  to the channel signal, where the two dimensional source is mapped to the nearest point on the space filling spiral.  The comparative performance results are shown in Figure \ref{res_21}. The proposed mapping outperforms the Archimedian spiral \cite{hekland} over the entire range of CSNR values. It is notable that the ``intermediate" option of only optimizing  the decoder captures a significant portion of the gains.

\subsection {Source-Channel Coding with Decoder Side Information}

In this section, we demonstrate the use of the proposed algorithm by focusing on the specific scenario of Figure 2. It must be emphasized that, while the algorithm is general and directly applicable to any choice of source and channel dimensions and distributions, for conciseness of the results section we will assume that sources are jointly Gaussian scalars with correlation coefficient $\rho$, where $|\rho|\leq 1$, and are identically distributed as described in Section II.B. We also assume that the noise is scalar and Gaussian.

Figure \ref{example1} presents a sample of encoding mappings obtained by varying the correlation coefficient and CSNR. Interestingly, the analog mapping captures the central characteristic observed in digital Wyner-Ziv mappings, in the sense of many-to-one mappings, where multiple source intervals are mapped to the same channel interval, which will potentially be resolved by the decoder given the side information. Within each bin, there is a mapping function which is approximately linear in this case (scalar Gaussian sources and channel). 
\begin{figure*}
  \centering
  \subfigure[ CSNR=10dB, $\rho=0.97$]{
			\includegraphics[scale=0.30]{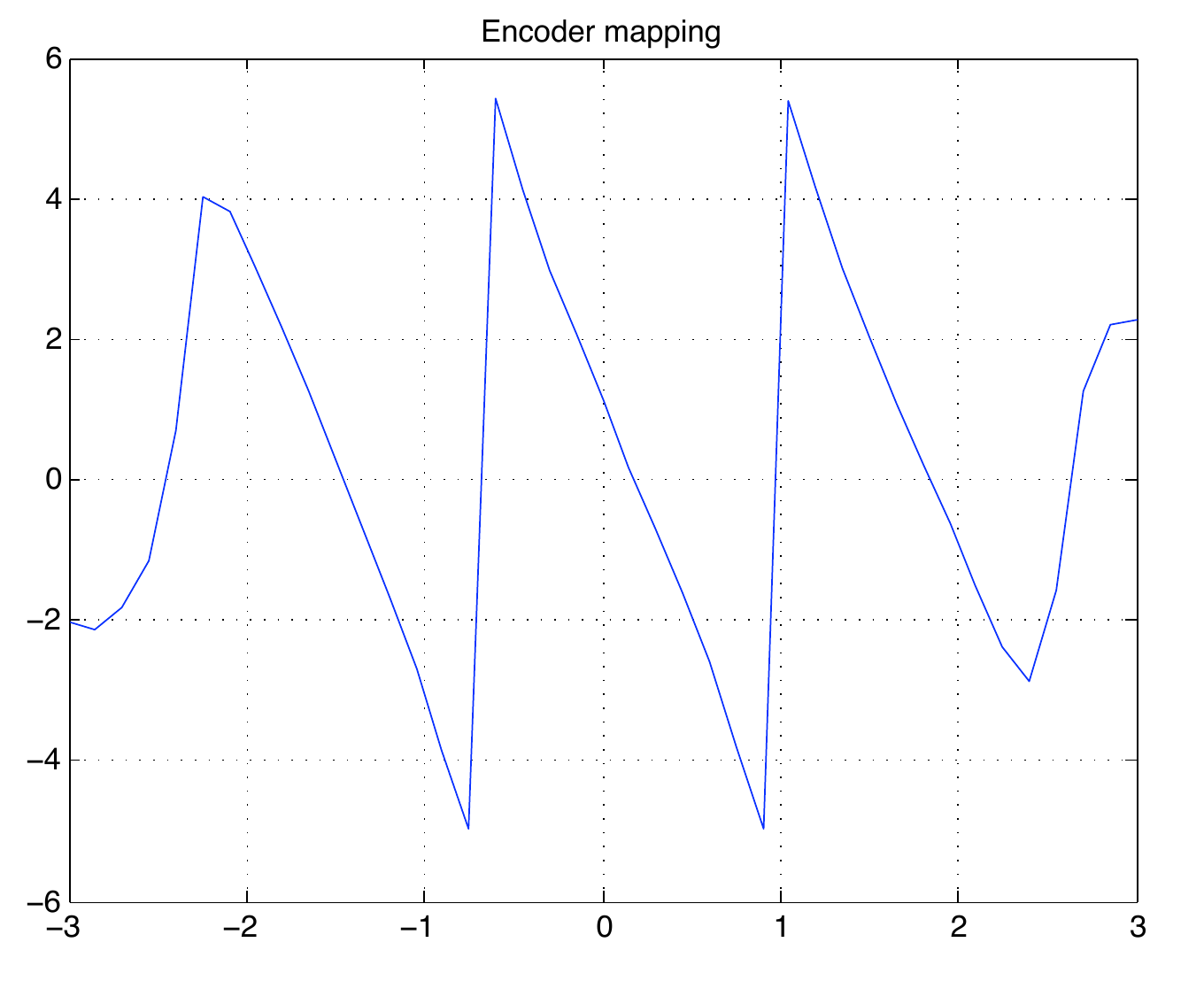}}
	               \subfigure[ CSNR=22dB, $\rho=0.97$]{
			\includegraphics[scale=0.30]{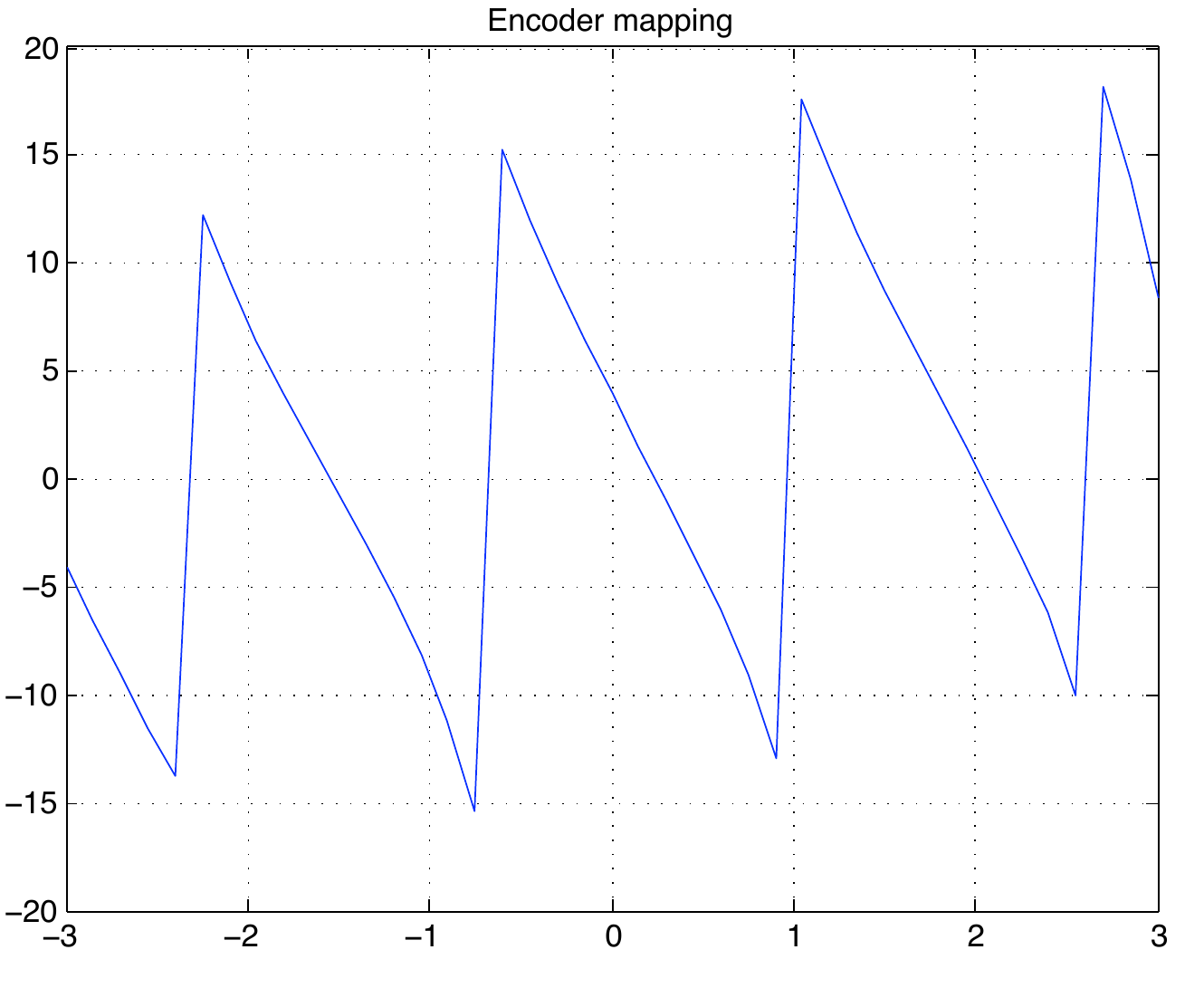}}
			\subfigure[ CSNR=10dB, $\rho=0.9$]{
			\includegraphics[scale=0.30]{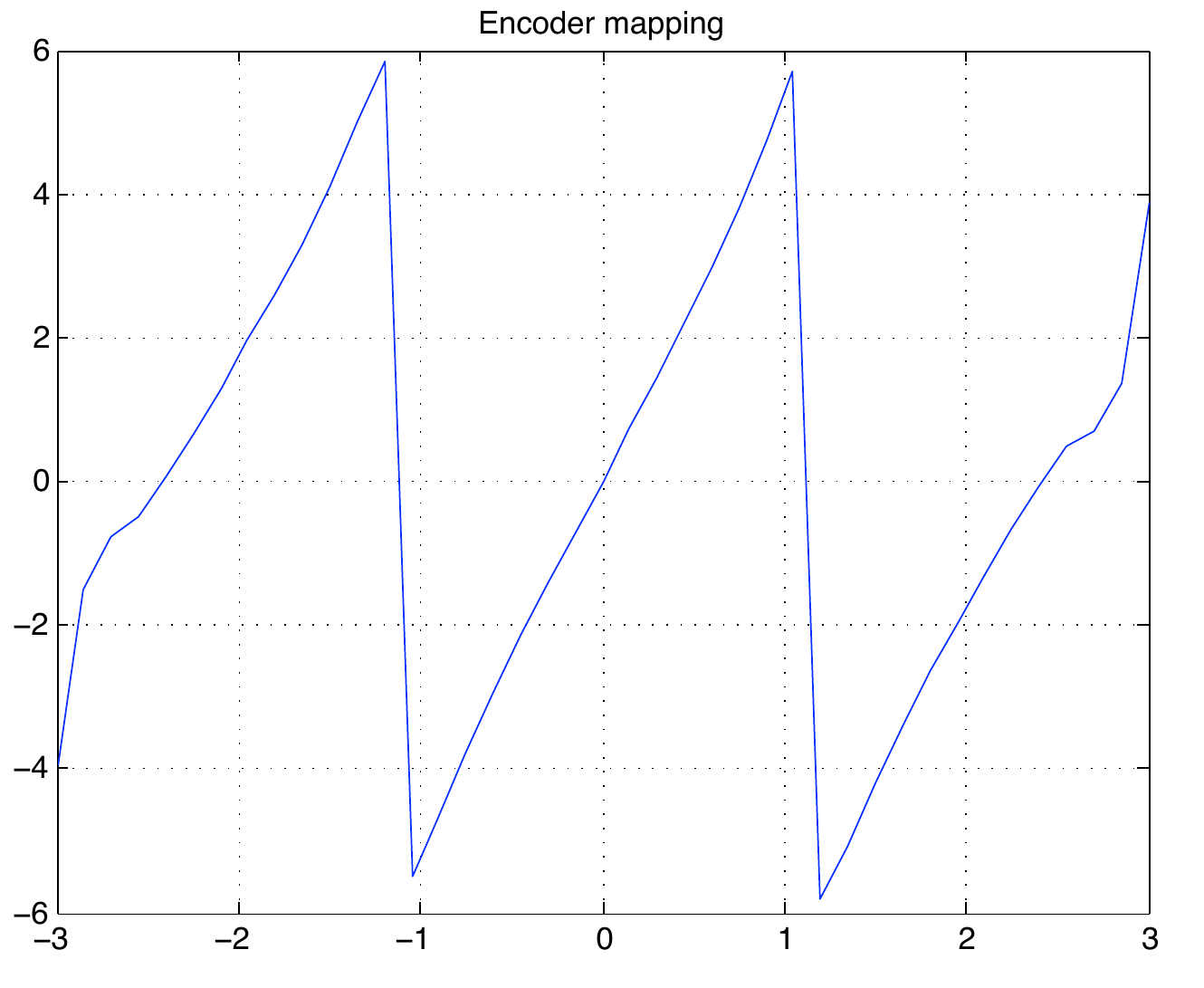}}
			\subfigure[ CSNR=22dB, $\rho=0.9$]{
			\includegraphics[scale=0.30]{ro9csnr10.pdf}}
\caption{Encoder mappings for Gaussian scalar source, channel and side information at different CSNR and correlation levels}
	\label{example1}
\end{figure*}
To see the effect of correlation on the encoding mappings, we note how the mapping changes as we lower the correlation from $\rho=0.97$ to $\rho=0.9$. As intuitively expected, the side information is less reliable and source points that are mapped to the same channel representation grow further apart from each other. 
\begin{figure}
\centering
			\includegraphics[scale=0.66]{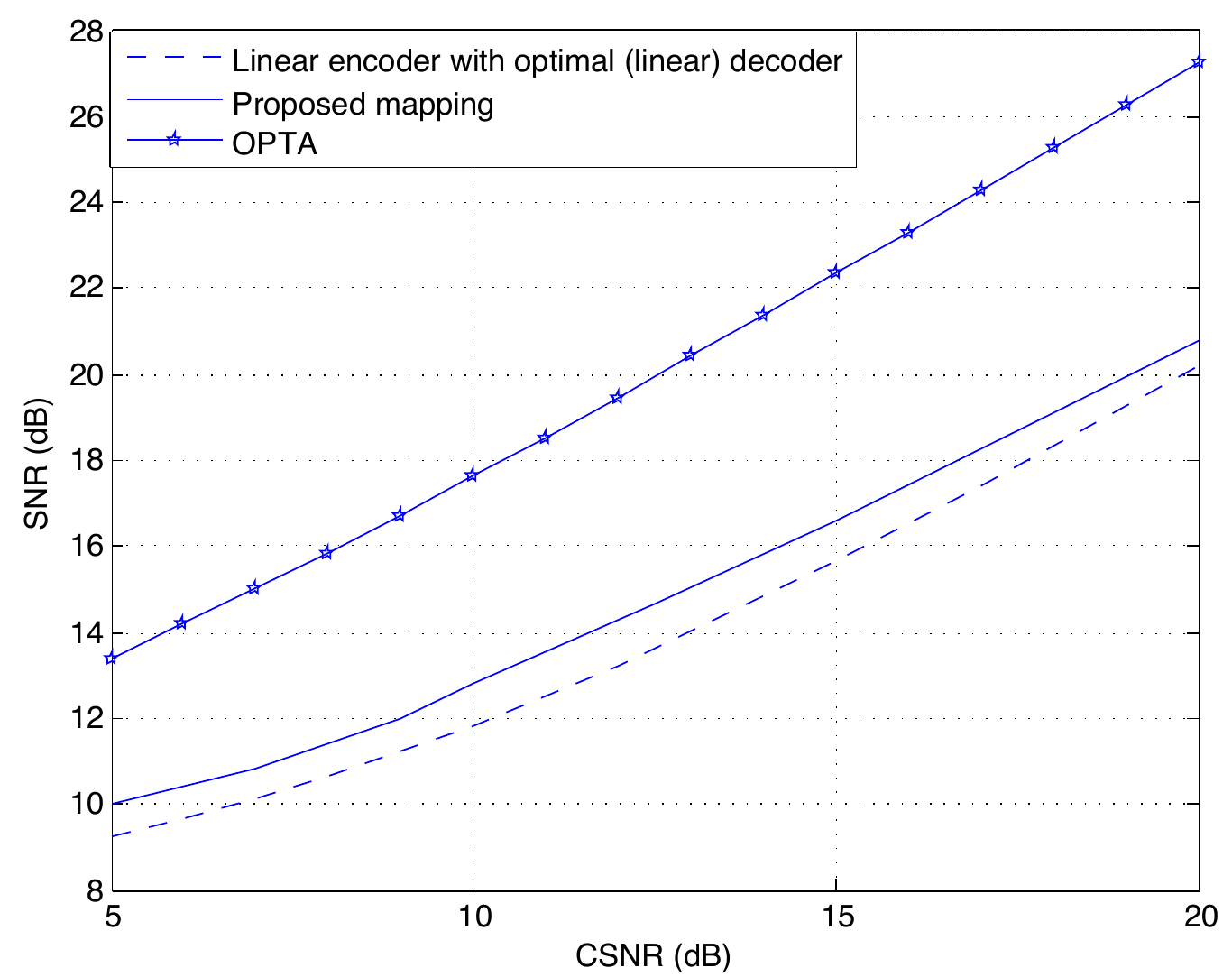}
			 \caption{Comparative results for correlation coefficient $\rho=0.9$, Gaussian scalar source, channel and side information}
\label{sideinfocomp}
\end{figure}
Comparative performance results are shown in Figure \ref{sideinfocomp}. The proposed mapping significantly outperforms linear mapping over the entire range of CSNR values. We note that this characteristic of the encoding mappings was also noted in experiments with  the PCCOVQ approach in \cite{karlsson2010optimized}, and was implemented in \cite{chen2011zero}, for optimizing hybrid (digital+analog) mappings.

\section{Discussion and Future Work}
In this paper, we studied the zero-delay source-channel coding problem. First, we derived the necessary conditions for optimality of the encoding and decoding mappings for a given source-channel system. Based on the necessary conditions, we proposed an iterative algorithm which generates  locally optimal encoder and decoder mappings. Comparative results and example mappings are provided and it is shown that the proposed method improves upon the results of prior work. The algorithm does not guarantee a globally optimal solution. This problem can be largely mitigated by using more powerful optimization, in particular a deterministic annealing approach \cite{da}.
Moreover, we investigated the functional properties of the zero-delay source-channel coding problem. We specialized to the scalar setting and showed that the problem is concave in the source and the channel noise densities and convex in the channel input density. Then, using these functional properties and the necessary conditions of optimality we had derived, we obtained the necessary and sufficient condition for linearity of optimal mappings. We also studied the implications of this matching condition and particularly showed that the optimal mappings converge to linear at asymptotically high CSNR for a Gaussian source, irrespective of the channel density and similarly for a Gaussian channel, at asymptotically low CSNR, irrespective of the source.

The numerical algorithm presented in this paper is feasible for relatively low source and channel dimensions ($m,k$). For high dimensional vector spaces, the numerical approach should be supported by imposing a tractable structure to the mappings, to mitigate the problem of the dimensionality. A set of preliminary results in this direction recently appeared in\cite{akyol2012linear}, where a linear transformation followed by scalar non-linear mappings were utilized for the decoder side information setting. The purely linear solution had been investigated in \cite{goela2012reduced}, where numerical algorithms are proposed to find the optimal bandwidth compression transforms in network settings.

The analysis in this paper, specifically conditions for linearity (and generalizations to other structural forms) of optimal mappings, as well as the numerical approach, can be extended to well known control problems such as the optimal jamming problem \cite{basar1983gaussian} and Witsenhausen's counterexample \cite{witsenhausen1968counterexample,basar2008variations}. 

An interesting question is on the existence of structure in the optimal mappings in some fundamental  scenarios. For instance, in \cite{chen2011zero}, a hybrid digital-analog encoding was employed for the problem of zero-delay source-channel coding with decoder side information, where the source, the side information and the channel noise are all scalar and Gaussian. The reported performance results are very close to the performance of the optimal unconstrained mappings. In contrast, in \cite{karlsson2010optimized}, sawtooth-like structure was assumed and its parameters were optimized as well as PCCOVQ was employed to obtain the non-structured mappings, where non-negligible performance difference between these approaches was reported. Hence, this fundamental question, on whether the optimal zero-delay mappings are structured for the scalar Gaussian side information setting, is currently open. This problem can numerically be approached by employing a powerful non-convex optimization tool, such as deterministic annealing \cite{da}, and this approach is currently under investigation.

\appendices

\section{Optimality Conditions for Coding with Decoder Side Information}

Let the encoder $\boldsymbol g(\cdot)$ be fixed. Then, the optimal decoder is the MMSE estimator  of $\boldsymbol X_1$ given $\boldsymbol X_2$ and $\boldsymbol {\hat Y}$:
\begin{equation}
\boldsymbol h(\boldsymbol{ \hat y},\boldsymbol x_2)= \mathbb E \{\boldsymbol X_1|\boldsymbol {\hat y},\boldsymbol x_2\}.
\end{equation}
Plugging the expressions for expectation, applying Bayes' rule and noting that $f_{\hat Y|X_1}({\boldsymbol { \hat y}, \boldsymbol x_1})= f_{Z}[{\boldsymbol { \hat y}-\boldsymbol g(\boldsymbol x_1)}]$, the optimal decoder can be written, in terms of known quantities, as    
        \begin{equation} \label{decoder_map}
     {\boldsymbol h(\boldsymbol {\hat y},\boldsymbol x_2)}= \dfrac{ \int {\boldsymbol x_1}  \,  f_{X_1,X_2}( {\boldsymbol x_1,\boldsymbol x_2})  \,  f_{Z} [{\boldsymbol {\hat y}-\boldsymbol g(\boldsymbol x_1)}]\, { \mathrm{d}\boldsymbol x_1} } {  \int  \,      f_{X_1,X_2}( {\boldsymbol x_1, \boldsymbol x_2})  \,  f_{Z} [{\boldsymbol {\hat y}-\boldsymbol g(\boldsymbol x_1)}]  \,  \mathrm{d}\boldsymbol x_1}.
      \end{equation}
 
To derive the necessary condition for optimality of $\boldsymbol g(\cdot)$, we consider the distortion functional   
\begin{equation}
D[{\boldsymbol g,\boldsymbol h}] =\mathbb E \{ || { [\boldsymbol X_1- \boldsymbol h(\boldsymbol g(\boldsymbol X_1)+\boldsymbol Z,\boldsymbol X_2)}] ||^2\},
\end{equation}
and construct the Lagrangian cost functional:
\begin{equation}\label{total_cost_SI}
J [{\boldsymbol g,\boldsymbol h}] =D[{\boldsymbol g,\boldsymbol h}]+\lambda P [{\boldsymbol g}] .
  \end{equation}
Now, let us assume the decoder $\boldsymbol h(\cdot)$ is fixed. To obtain necessary conditions we apply the standard method in variational calculus:
\begin{equation}\label {encoder_map}
 \nabla_{\boldsymbol g} J [{\boldsymbol g},\boldsymbol h] (\boldsymbol x_1,\boldsymbol x_2)=0, \, \, \forall \, \boldsymbol x_1, \,\boldsymbol  x_2,
\end{equation} 
where 
 \begin{align}
 \nabla_{\boldsymbol g } J[\boldsymbol g,\boldsymbol h] (\boldsymbol x_1,\boldsymbol x_2)\!= \lambda f_{X_1,X_2}( {\boldsymbol  x_1,\boldsymbol x_2})   {\boldsymbol g(\boldsymbol x_1)} \! -\!\int \!  \boldsymbol h'(\boldsymbol g(\boldsymbol x_1)\!+\!\boldsymbol z,\boldsymbol x_2)     \left [{\boldsymbol x\!-\!\boldsymbol h(\boldsymbol g(\boldsymbol x)\!+\!\boldsymbol z,\boldsymbol x_2)}\right]  f_Z({\boldsymbol z}) f_{X_1,X_2} ({\boldsymbol  x_1,\boldsymbol x_2})   \mathrm{d}\boldsymbol z.
\end{align}

\bibliographystyle{IEEEbib}
\bibliography{ref}

\begin{thebibliography}{10}

\bibitem{goblick}
T.~Goblick~Jr,
\newblock ``{Theoretical limitations on the transmission of data from analog
  sources},''
\newblock {\em IEEE Transactions on Information Theory}, vol. 11, no. 4, pp.
  558--567, 1965.

\bibitem{coverbook}
T.M. Cover and J.A. Thomas,
\newblock {\em {Elements of information theory}},
\newblock J.Wiley New York, 1991.

\bibitem{mittal2002hybrid}
U.~Mittal and N.~Phamdo,
\newblock ``{Hybrid digital-analog (HDA) joint source-channel codes for
  broadcasting and robust communications},''
\newblock {\em IEEE Transactions on Information Theory}, vol. 48, no. 5, pp.
  1082--1102, 2002.

\bibitem{skoglund2006hybrid}
M.~Skoglund, N.~Phamdo, and F.~Alajaji,
\newblock ``{Hybrid digital-analog source-channel coding for bandwidth
  compression/expansion},''
\newblock {\em IEEE Trans. Information Theory}, vol. 52, no. 8, pp. 3757--3763,
  2006.

\bibitem{shannon1949cpn}
CE~Shannon,
\newblock ``{Communication in the presence of noise},''
\newblock {\em Proceedings of the IRE}, vol. 37, no. 1, pp. 10--21, 1949.

\bibitem{kotelnikov}
VA~Kotelnikov,
\newblock {\em {The theory of optimum noise immunity}},
\newblock New York: McGraw-Hill, 1959.

\bibitem{fuldseth}
A.~Fuldseth and TA~Ramstad,
\newblock ``{Bandwidth compression for continuous amplitude channels based on
  vector approximation to a continuous subset of the source signal space},''
\newblock in {\em Proceedings of the IEEE International Conference on
  Acoustics, Speech, and Signal Processing,}, 1997, vol.~4.

\bibitem{chung}
S.Y. Chung,
\newblock {\em {On the construction of some capacity approaching coding
  schemes}},
\newblock Ph.D. thesis, Massachusetts Institute of Technology, 2000.

\bibitem{vaishampayan2003curves}
V.A. Vaishampayan and S.I.R. Costa,
\newblock ``Curves on a sphere, shift-map dynamics, and error control for
  continuous alphabet sources,''
\newblock {\em IEEE Transactions on Information Theory}, vol. 49, no. 7, pp.
  1658--1672, 2003.

\bibitem{ramstad}
T.A. Ramstad,
\newblock ``{Shannon mappings for robust communication},''
\newblock {\em Telektronikk}, vol. 98, no. 1, pp. 114--128, 2002.

\bibitem{hekland2009}
F.~Hekland, P.A. Floor, and T.A. Ramstad,
\newblock ``{Shannon-Kotelnikov mappings in joint source-channel coding},''
\newblock {\em IEEE Transactions on Communications}, vol. 57, no. 1, pp.
  94--105, 2009.

\bibitem{hu2011analog}
Y.~Hu, J.~Garcia-Frias, and M.~Lamarca,
\newblock ``Analog joint source-channel coding using non-linear curves and mmse
  decoding,''
\newblock {\em IEEE Transactions on Communications}, vol. 59, no. 11, pp.
  3016--3026, 2011.

\bibitem{wernersson2009polynomial}
N.~Wernersson, M.~Skoglund, and T.~Ramstad,
\newblock ``{Polynomial based analog source channel codes},''
\newblock {\em IEEE Transactions on Communications,}, vol. 57, no. 9, pp.
  2600--2606, 2009.

\bibitem{floor2007power}
P.A. Floor, T.A. Ramstad, and N.~Wernersson,
\newblock ``Power constrained channel optimized vector quantizers used for
  bandwidth expansion,''
\newblock in {\em Proceedings of the 4th International Symposium on Wireless
  Communication Systems, 2007.} IEEE, 2007, pp. 667--671.

\bibitem{karlsson2010optimized}
J.~Karlsson and M.~Skoglund,
\newblock ``{Optimized low delay source channel relay mappings},''
\newblock {\em IEEE Transactions on Communications}, vol. 58, no. 5, pp.
  1397--1404, 2010.

\bibitem{olc}
K.H. Lee and D.~Petersen,
\newblock ``{Optimal linear coding for vector channels},''
\newblock {\em IEEE Transactions on Communications}, vol. 24, no. 12, pp.
  1283--1290, 1976.

\bibitem{basar1980performance}
T.~Ba\c{s}ar, B.~Sankur, and H.~Abut,
\newblock ``Performance bounds and optimal linear coding for discrete-time
  multichannel communication systems,''
\newblock {\em IEEE Transactions on Information Theory}, vol. 26, no. 2, pp.
  212--217, 1980.

\bibitem{fine1964properties}
T.~Fine,
\newblock ``{Properties of an optimum digital system and applications},''
\newblock {\em IEEE Transactions on Information Theory}, vol. 10, no. 4, pp.
  287--296, 1964.

\bibitem{ziv}
J.~Ziv,
\newblock ``{The behavior of analog communication systems},''
\newblock {\em IEEE Transactions on Information Theory,}, vol. 16, no. 5, pp.
  587--594, 1970.

\bibitem{trott}
MD~Trott,
\newblock ``{Unequal error protection codes: Theory and practice},''
\newblock in {\em Proceedings of the IEEE Information Theory Workshop}, 1996,
  p.~11.

\bibitem{emrah_itw10}
E.~Akyol, K.~Rose, and TA~Ramstad,
\newblock ``{Optimal mappings for joint source channel coding},''
\newblock in {\em Proceedings of the IEEE Information Theory Workshop}. IEEE,
  2010, pp. 1--5.

\bibitem{emrah_dcc10}
E.~Akyol, K.~Rose, and TA~Ramstad,
\newblock ``{Optimized analog mappings for distributed source channel
  coding},''
\newblock in {\em Proceedings of the IEEE Data Compression Conference}. IEEE,
  2010, pp. 159--168.

\bibitem{witsenhausen1979structure}
HS~Witsenhausen,
\newblock ``On the structure of real-time source coders,''
\newblock {\em Bell Syst. Tech. J}, vol. 58, no. 6, pp. 1437--1451, 1979.

\bibitem{walrand1983optimal}
J.~Walrand and P.~Varaiya,
\newblock ``Optimal causal coding-decoding problems,''
\newblock {\em IEEE Transactions on Information Theory}, vol. 29, no. 6, pp.
  814--820, 1983.

\bibitem{teneketzis2006structure}
D.~Teneketzis,
\newblock ``On the structure of optimal real-time encoders and decoders in
  noisy communication,''
\newblock {\em IEEE Transactions on Information Theory}, vol. 52, no. 9, pp.
  4017--4035, 2006.

\bibitem{yuksel2010optimal}
S.~Yuksel,
\newblock ``On optimal causal coding of partially observed {M}arkov sources in
  single and multi-terminal settings,''
\newblock {\em IEEE Transactions on Information Theory}, vol. PP, no. 99, pp.
  1, 2012.

\bibitem{witsenhausen1968counterexample}
H.S. Witsenhausen,
\newblock ``A counterexample in stochastic optimum control,''
\newblock {\em SIAM Journal on Control}, vol. 6, no. 1, pp. 131--147, 1968.

\bibitem{basar2008variations}
T.~Ba\c{s}ar,
\newblock ``Variations on the theme of the {W}itsenhausen counterexample,''
\newblock in {\em Proceedings of the IEEE Conference on Decision and Control
  )}. IEEE, 2008, pp. 1614--1619.

\bibitem{emrah_estimation}
E.~Akyol, K.~Viswanatha, and K.~Rose,
\newblock ``On conditions for linearity of optimal estimation,''
\newblock {\em IEEE Transactions on Information Theory}, vol. 58, no. 6, pp.
  3497 --3508, 2012.

\bibitem{slepianwolf}
D.~Slepian and J.~Wolf,
\newblock ``{Noiseless coding of correlated information sources},''
\newblock {\em IEEE Transactions on Information Theory}, vol. 19, no. 4, pp.
  471--480, 1973.

\bibitem{wynerziv}
A.~Wyner and J.~Ziv,
\newblock ``{The rate-distortion function for source coding with side
  information at the decoder},''
\newblock {\em IEEE Transactions on Information Theory}, vol. 22, no. 1, pp.
  1--10, 1976.

\bibitem{dlb}
A.~Ingber, I.~Leibowitz, R.~Zamir, and M.~Feder,
\newblock ``{Distortion lower bounds for finite dimensional joint
  source-channel coding},''
\newblock in {\em Proceedings of the IEEE International Symposium on
  Information Theory}, 2008, pp. 1183--1187.

\bibitem{tridenski2011bounds}
S.~Tridenski and R.~Zamir,
\newblock ``Bounds for joint source-channel coding at high snr,''
\newblock in {\em Proceedings of the IEEE International Symposium on
  Information Theory}. IEEE, 2011, pp. 771--775.

\bibitem{reani2012data}
A.~Reani and N.~Merhav,
\newblock ``Data processing lower bounds for scalar lossy source codes with
  side information at the decoder,''
\newblock in {\em Proceedings of the IEEE International Symposium on
  Information Theory}. IEEE, 2012, pp. 1--5.

\bibitem{wu2012functional}
Y.~Wu and S.~Verd{\'u},
\newblock ``Functional properties of minimum mean-square error and mutual
  information,''
\newblock {\em IEEE Transactions on Information Theory}, vol. 58, no. 3, pp.
  1289--1301, 2012.

\bibitem{shannon1949mathematical}
C.E. Shannon,
\newblock ``The mathematical theory of information,''
\newblock {\em Bell System Technical Journal}, vol. 27, no. 6, pp. 379 --423,
  1949.

\bibitem{csiszar2011information}
I.~Csiszar and J.~K{\"o}rner,
\newblock {\em Information theory: {C}oding theorems for discrete memoryless
  systems},
\newblock Cambridge Univ Pr, 2011.

\bibitem{wu2011witsenhausen}
Y.~Wu and S.~Verd{\'u},
\newblock ``WitsenhausenÕs counterexample: A view from optimal transport
  theory,''
\newblock in {\em Proceedings of the IEEE Conference on Decision and Control},
  2011.

\bibitem{villani2009optimal}
C.~Villani,
\newblock {\em Optimal transport: old and new}, vol. 338,
\newblock Springer Verlag, 2009.

\bibitem{Luenberger}
D.G. Luenberger,
\newblock {\em {Optimization by Vector Space Methods}},
\newblock John Wiley \& Sons Inc., 1969.

\bibitem{gadkari1999robust}
S.~Gadkari and K.~Rose,
\newblock ``{Robust vector quantizer design by noisy channel relaxation},''
\newblock {\em IEEE Transactions on Communications}, vol. 47, no. 8, pp.
  1113--1116, 1999.

\bibitem{Knagenhjelm}
P.~Knagenhjelm,
\newblock ``{A recursive design method for robust vector quantization},''
\newblock in {\em Proceedings of the International Conference Signal Processing
  Applications and Technology}, 1992, pp. 948--954.

\bibitem{dudley2002real}
R.M. Dudley,
\newblock {\em {Real Analysis and Probability}},
\newblock Cambridge University Press, 2002.

\bibitem{billingsley2008probability}
P.~Billingsley,
\newblock {\em {Probability and Measure}},
\newblock John Wiley \& Sons Inc., 2008.

\bibitem{hekland}
F.~Hekland, GE~Oien, and TA~Ramstad,
\newblock ``{Using 2: 1 Shannon mapping for joint source-channel coding},''
\newblock in {\em Proceedings of the IEEE Data Compression Conference}, 2005,
  pp. 223--232.

\bibitem{chen2011zero}
X.~Chen and E.~Tuncel,
\newblock ``Zero-delay joint source-channel coding for the {G}aussian
  {W}yner-{Z}iv problem,''
\newblock in {\em Proceedings of the IEEE International Symposium on
  Information Theory}, 2011, pp. 2929--2933.

\bibitem{da}
K.~Rose,
\newblock ``{Deterministic annealing for clustering, compression,
  classification, regression, and related optimization problems},''
\newblock {\em Proceedings of the IEEE}, vol. 86, no. 11, pp. 2210--2239, 1998.

\bibitem{akyol2012linear}
E.~Akyol and K.~Rose,
\newblock ``On linear transforms in zero-delay {G}aussian source channel
  coding,''
\newblock in {\em Proceedings of the IEEE International Symposium on
  Information Theory}. IEEE, 2012, pp. 1543--1547.

\bibitem{goela2012reduced}
N.~Goela and M.~Gastpar,
\newblock ``Reduced-dimension linear transform coding of correlated signals in
  networks,''
\newblock {\em IEEE Transactions on Signal Processing}, vol. 60, no. 6, pp.
  3174--3187, 2012.

\bibitem{basar1983gaussian}
T.~Ba\c{s}ar,
\newblock ``The {G}aussian test channel with an intelligent jammer,''
\newblock {\em IEEE Transactions on Information Theory}, vol. 29, no. 1, pp.
  152--157, 1983.

\end{thebibliography}

%
\end{document}